\documentclass[english]{article}
\usepackage[T1]{fontenc}
\usepackage[latin9]{inputenc}
\usepackage{float}
\usepackage{amsmath}
\usepackage{amsthm}
\usepackage{amssymb}
\usepackage{graphicx}

\makeatletter
\theoremstyle{plain}
\newtheorem{thm}{\protect\theoremname}
\theoremstyle{definition}
\newtheorem{defn}[thm]{\protect\definitionname}
\theoremstyle{remark}
\newtheorem{rem}[thm]{\protect\remarkname}
\theoremstyle{plain}
\newtheorem{lem}[thm]{\protect\lemmaname}
\theoremstyle{plain}
\newtheorem{fact}[thm]{\protect\factname}
\theoremstyle{plain}
\newtheorem{prop}[thm]{\protect\propositionname}

\usepackage{fullpage}
\usepackage{hyperref}
\usepackage{centernot}
\usepackage[bottom]{footmisc}
\hypersetup{
  colorlinks   = true, 
  urlcolor     = red, 
  linkcolor    = blue, 
  citecolor   = blue 
}
\date{}
\@ifundefined{showcaptionsetup}{}{%
 \PassOptionsToPackage{caption=false}{subfig}}
\usepackage{subfig}
\makeatother

\usepackage{babel}
\providecommand{\definitionname}{Definition}
\providecommand{\factname}{Fact}
\providecommand{\lemmaname}{Lemma}
\providecommand{\propositionname}{Proposition}
\providecommand{\remarkname}{Remark}
\providecommand{\theoremname}{Theorem}

\begin{document}
\global\long\def\goinf{\rightarrow\infty}%
\global\long\def\gozero{\rightarrow0}%
\global\long\def\bra{\langle}%
\global\long\def\ket{\rangle}%
\global\long\def\union{\cup}%
\global\long\def\intersect{\cap}%
\global\long\def\abs#1{\left|#1\right|}%
\global\long\def\norm#1{\left\Vert #1\right\Vert }%
\global\long\def\floor#1{\left\lfloor #1\right\rfloor }%
\global\long\def\ceil#1{\left\lceil #1\right\rceil }%
\global\long\def\expect{\mathbb{E}}%
\global\long\def\e{\mathbb{E}}%
\global\long\def\r{\mathbb{R}}%
\global\long\def\p{\mathbb{P}}%
\global\long\def\n{\mathbb{N}}%
\global\long\def\q{\mathbb{Q}}%
\global\long\def\c{\mathbb{C}}%
\global\long\def\z{\mathbb{Z}}%
\global\long\def\grad{\nabla}%
\global\long\def\t{^{\prime}}%
\global\long\def\all{\forall}%
\global\long\def\eps{\varepsilon}%
\global\long\def\quadvar#1{V_{2}^{\pi}\left(#1\right)}%
\global\long\def\cal#1{\mathcal{#1}}%
\global\long\def\cross{\times}%
\global\long\def\del{\nabla}%
\global\long\def\parx#1{\frac{\partial#1}{\partial x}}%
\global\long\def\pary#1{\frac{\partial#1}{\partial y}}%
\global\long\def\parz#1{\frac{\partial#1}{\partial z}}%
\global\long\def\part#1{\frac{\partial#1}{\partial t}}%
\global\long\def\partheta#1{\frac{\partial#1}{\partial\theta}}%
\global\long\def\parr#1{\frac{\partial#1}{\partial r}}%
\global\long\def\curl{\nabla\times}%
\global\long\def\rotor{\nabla\times}%
\global\long\def\one{\mathbf{1}}%
\global\long\def\Hom{\text{Hom}}%
\global\long\def\pr#1{\text{Pr}\left[#1\right]}%
\global\long\def\almost{\mathbf{\approx}}%
\global\long\def\tr{\text{Tr}}%
\global\long\def\var{\mathrm{Var}}%
\global\long\def\cov{\text{Cov}}%
\global\long\def\onenorm#1{\left\Vert #1\right\Vert _{1}}%
\global\long\def\twonorm#1{\left\Vert #1\right\Vert _{2}}%
\global\long\def\Inj{\mathfrak{Inj}}%
\global\long\def\inj{\mathsf{inj}}%
\global\long\def\lone{\mathrm{L}^{1}}%
\global\long\def\ltwo{\mathrm{L}^{2}}%
\global\long\def\elone{\mathrm{L}^{1}}%
\global\long\def\eltwo{\mathrm{L}^{2}}%
\global\long\def\Inf{\mathrm{Inf}}%
\global\long\def\i{\mathrm{I}}%
\global\long\def\sign{\mathrm{sign}}%
\global\long\def\tensor{\otimes}%
\global\long\def\ans{\mathrm{ans}}%
\global\long\def\axis{\Delta_{\dagger}}%

\global\long\def\g{\mathfrak{\cal G}}%
\global\long\def\f{\mathfrak{\cal F}}%

\title{Noise sensitivity from fractional query algorithms and the axis-aligned
Laplacian}
\author{Renan Gross\thanks{Weizmann Institute of Science. Email: renan.gross@weizmann.ac.il.
Supported by the Adams Fellowship Program of the Israel Academy of
Sciences and Humanities.}}
\maketitle
\begin{abstract}
We introduce the notion of classical fractional query algorithms,
which generalize decision trees in the average-case setting, and can
potentially perform better than them. We show that the limiting run-time
complexity of a natural class of these algorithms obeys the non-linear
partial differential equation $\min_{k}\partial^{2}u/\partial x_{k}^{2}=-2$,
and that the individual bit revealment satisfies the Schramm-Steif
bound for Fourier weight, connecting noise sensitivity with PDEs.
We discuss relations with other decision tree results.
\end{abstract}
\tableofcontents{}

\section{Introduction}

\subsection{Decision trees and Boolean functions}

A decision tree is an adaptive algorithm for determining the value
of a function $f:\left\{ -1,1\right\} ^{n}\to\r$ given an unknown
input $x\in\left\{ -1,1\right\} ^{n}$. At each step, the algorithm
(possibly randomly) chooses an index $i\in\left[n\right]$, and queries
the value of the bit $x_{i}$. We do not require that the algorithm
always calculate $f\left(x\right)$ exactly: it can stop running and
output some value in $\r$ even if it has not queried enough bits
to fix the value of $f\left(x\right)$. In general, the goal of the
decision tree is to read as few bits as possible. Due to their simplicity
as a computational model, decision trees have been studied extensively,
especially in the \emph{worst-case} setting, where the complexity
of a tree is defined as the maximum expected number of queries it
makes over all inputs $x\in\left\{ -1,1\right\} ^{n}$; see the excellent
survey by Buhrman and de Wolf \cite{buhrman_de_wolf_survey} for an
exposition. 

In this paper we will investigate the \emph{average-case} setting,
where the input $x$ is drawn from some distribution $\mu$ on $\left\{ -1,1\right\} ^{n}$.
Two common complexity measures in this case are the expected number
of queries that the tree makes, and the maximum probability that it
reads any particular bit. In this setting, the complexity of a decision
tree is tied with central notions in the analysis of Boolean functions,
such as their variance and influences (see Section \ref{subsec:background_boolean}
for a review of Boolean functions). One result in this vein is given
by Schramm and Steif \cite{ss_quantitative_noise_sensitivity}, and
was originally used to show quantitative noise sensitivity for percolation
crossing events. It relates the probability of the algorithm to query
a bit (called \emph{revealment})\emph{ }to the Fourier mass at level
$k$ of the function:
\begin{thm}[Theorem 1.8 in \cite{ss_quantitative_noise_sensitivity}]
\label{thm:ss_inequality}Let $f:\left\{ -1,1\right\} ^{n}\to\r$
have Fourier representation $f\left(x\right)=\sum_{S\subseteq\left[n\right]}\hat{f}\left(S\right)\prod_{i\in S}x_{i}$,
and let $T$ be a decision tree calculating $f$ when the input is
uniform. Set $\delta=\max_{i}\p\left[\text{\ensuremath{T} reads bit \ensuremath{i}}\right]$.
Then for every $k\in\n$, the Fourier coefficients of $f$ satisfy
\begin{equation}
\sum_{\abs S=k}\hat{f}\left(S\right)^{2}\leq\delta k\norm f_{2}^{2}.\label{eq:ss_10_original}
\end{equation}
\end{thm}

Two vectors $x,y\in\left\{ -1,1\right\} ^{n}$ are said to be $\rho$-correlated
if $x$ and $y$ are uniform in $\left\{ -1,1\right\} ^{n}$ and $\e\left[x_{i}y_{i}\right]=\rho$
for all $i\in\left[n\right]$. A sequence of functions $f_{m}:\left\{ -1,1\right\} ^{n_{m}}\to\left\{ -1,1\right\} $
is said to be noise sensitive with respect to noise $\eps_{m}>0$,
if, when $x,y\in\left\{ -1,1\right\} ^{n_{m}}$ are $\left(1-\eps_{m}\right)$-correlated,
then $\lim_{m\to\infty}\e\left[f\left(x\right)f\left(y\right)\right]-\e\left[f\left(x\right)\right]^{2}=0$.
For monotone functions, if $\delta\left(f_{m}\right)\to0$, then Theorem
\ref{thm:ss_inequality}, together with a theorem of Benjamini, Kalai
and Schramm connecting Fourier coefficients and noise sensitivity
\cite[Theorem 1.5]{bks_noise_sensitivity_of_boolean_functions}, gives
quantitative bounds on how small the noise $\eps_{m}$ can be, depending
on how quickly $\delta\left(f_{m}\right)\to0$. 

Another useful inequality was given by O'Donnell, Saks, Schramm and
Steif \cite{osss_every_decision_tree_has_an_influential_variable},
and relates the revealment probabilities to the function's influences.
The influence $\Inf_{i}\left(f\right)$ of the $i$-th bit is the
probability that flipping the $i$-th bit changes $f$'s value when
the input is uniform. The theorem states:
\begin{thm}[Theorem 1.1 in \cite{osss_every_decision_tree_has_an_influential_variable}]
\label{thm:osss_inequality}Let $f:\left\{ -1,1\right\} ^{n}\to\left\{ -1,1\right\} $.
Let $T$ be a decision tree calculating $f$ when the input is uniform
and set $\delta_{i}=\p\left[\text{\ensuremath{T} reads bit \ensuremath{i}}\right]$.
Then 
\begin{equation}
\var\left(f\right)\leq\sum_{i=1}^{n}\delta_{i}\Inf_{i}\left(f\right).\label{eq:original_osss_inequality}
\end{equation}
\end{thm}

Inequality (\ref{eq:original_osss_inequality}) can be used to show
that functions which have small decision trees (i.e. make a small
number of queries) must have an influential variable. 

The above two theorems show why it is beneficial to find decision
trees which query each bit with as small probability as possible.
However, a classical theorem by Benjamini, Schramm and Wilson \cite{bsw_balanced_boolean_functions}
states that the revealment $\delta$ cannot be too small.
\begin{thm}[Theorem 2 in \cite{bsw_balanced_boolean_functions}]
\label{thm:bsw_inequality}Let $f:\left\{ -1,1\right\} ^{n}\to\left\{ -1,1\right\} $.
Let $T$ be a decision tree calculating $f$ which can err with probability
$w$. Let $\delta=\max_{i}\p\left[\text{\ensuremath{T} reads bit \ensuremath{i}}\right]$.
Then 
\[
w\geq\frac{1}{8}\var\left(f\right)-\frac{1}{4}n\delta^{2}.
\]
\end{thm}

\subsection{Our results}

Our main results are summarized as follows.
\begin{itemize}
\item Generalizing the notion decision trees, we define \emph{fractional
query algorithms}, where a query does not return the value of the
bit, but rather reveals some information on what the bit is likely
to be. Repeated fractional queries reveal more information, turning
the bits into time-dependent processes $X_{i}\left(t\right)$ taking
values in $\left[-1,1\right]$. This class of algorithms contains
the class of decision trees, but can potentially contain algorithms
which have better run-time complexity. Our framework allows us to
formulate the problem of decision tree complexity in the language
of optimal control theory.
\item We prove Theorem \ref{thm:ss_inequality} for fractional query algorithms,
and show that weaker versions of Theorem \ref{thm:osss_inequality}
and Theorem \ref{thm:bsw_inequality} also apply.
\item We show that for a natural class of fractional query algorithms, the
limiting run-time complexity obeys a partial differential equation
corresponding to a simple dynamic programming principle. 
\end{itemize}

\subsubsection{Fractional query algorithms}

Let $\mu$ be a product measure on $\left\{ -1,1\right\} ^{n}$. There
are two equivalent views for average-case decision tree inputs. In
the first, an input $x$ is chosen according to $\mu$, and the decision
tree simply queries its bits. In the second, no input is chosen beforehand,
and whenever the decision tree queries a bit, that bit is randomly
set to $\pm1$ according to $\mu$. In this case, the input to the
decision tree can be seen as a stochastic process $X\left(t\right)\in\left\{ -1,*,1\right\} ^{n}$.
The process starts at $X\left(0\right)=\left(*,\ldots,*\right)$,
where the $*$'s indicate that none of the bits have been read yet.
At time $t$, the decision tree picks a (possibly random) index $i_{t}$
to query, and the $i_{t}$-th coordinate is set to $\pm1$ with probabilities
according to $\mu$. The process stops at time $\tau$, once the partial
input $X\left(\tau\right)$ determines $f$ (or once the algorithm
guesses the value of $f$, if it is allowed to make errors). The probability
of reading a bit is given by $\e\left[X_{i}\left(\tau\right)^{2}\right]$
($*$'s are treated as $0$), and the expected number of queries $C$
is equal to $C=\e\sum_{i=1}^{n}X_{i}\left(\tau\right)^{2}$. 

Building on the second view, in our framework the random input bits
do not have to be completely queried, but rather can be queried fractionally.
At time $t$, the algorithm chooses a bit $i_{t}$ and makes an $\eps$-query
on it; this causes the $i_{t}$-th coordinate to be randomly updated:
$X_{i_{t}}\left(t+1\right)=X_{i_{t}}\left(t\right)\pm\eps$, each
with probability $1/2$. Thus, the input stochastic process $X\left(t\right)$
does not have to take values in $\left\{ -1,*,1\right\} ^{n}$ as
is the case for decision trees, but rather can take arbitrary values
in the continuous cube, $X\left(t\right)\in\left[-1,1\right]^{n}$.
Like ordinary decision trees, the process continues in this way until
some time $\tau$, when the algorithm decides that it has enough information
to guess the value of $f$ from the partial input $X\left(\tau\right)$.
The idea behind this type of algorithm is that by using partial queries,
the algorithm can get a sense of what the input will be, and avoid
completely revealing bits which are likely to be useless. When the
cost of making a partial query is appropriately defined, the complexity
of fractional query algorithms yields bounds on properties of Boolean
functions in a similar fashion to decision trees. We give the formal
definitions below.
\begin{defn}[Axis-aligned jump process]
\label{def:axis_aligned_process}Let $\eps=2^{-k}$ for some integer
$k>0$. An \emph{axis-aligned jump process} with jump size $\eps$
is a discrete-time martingale $X\left(t\right)\in\left[-1,1\right]^{n}\intersect2^{-k}\z^{n}$
defined as follows\footnote{Forcing the jumps to be of size $\eps$ makes the analysis of these
algorithms simpler. Allowing the jumps to be arbitrarily small moves
the analysis into the continuous time domain, which introduces problems
of measurability. See also Section \ref{subsec:open_questions}. We
will later take the limit of $\eps\to0$, in a controlled fashion.}. At time $t$, if $X\left(t\right)\notin\left\{ -1,1\right\} ^{n}$,
a direction $i_{t}\in\left[n\right]$ with $X_{i_{t}}\left(t\right)\in\left(-1,1\right)$
is chosen according to a\emph{ direction choosing strategy} $S$ (see
below for more details). The new position $X\left(t+1\right)$ is
then updated in a martingale fashion: for all $j\neq i_{t}$, we have
$X_{j}\left(t+1\right)=X_{j}\left(t\right)$, and for $i_{t}$ we
have $X_{i_{t}}\left(t+1\right)=X_{i_{t}}\left(t+1\right)\pm\eps$,
each with probability $1/2$. 

If $X\left(t\right)\in\left\{ -1,1\right\} ^{n}$, then $X\left(t\right)$
stays put, i.e. $X\left(t+1\right)=X\left(t\right)$. We denote this
final value as $X\left(\infty\right)$. It is not hard to see that
$X\left(\infty\right)$ distributes according to the product measure
whose mean is $X\left(0\right)$.
\end{defn}

\begin{rem}
There are several types of direction choosing strategies. The simplest
type are deterministic Markov strategies, where the direction in which
to go is only determined by the current position. We can thus write
$i_{t}=S\left(X\left(t\right)\right)$, for some function $S:\left[-1,1\right]^{n}\to\left[n\right]$.
In the most general setting, we can let the chosen direction depend
explicitly on the time $t$, on the history of the trajectory, and
on additional randomness which is independent of future decisions.
We can also allow ``lazy'' strategies, which sometimes do not move
$X\left(t\right)$ at all (this will be useful for later analysis). 
\end{rem}

\begin{defn}[Axis-aligned query algorithm]
\label{def:axis_aligned_algorithms}Let $X\left(t\right)$ be an
axis-aligned jump process with $X\left(0\right)=x_{0}$. An \emph{axis-aligned
query algorithm} is a triple $Q=\left(X\left(t\right),\tau,A\right)$,
where $\tau$ is an $X\left(t\right)$-adapted stopping time, and
$A:\left[-1,1\right]^{n}\to\r$ is some function. The output of the
algorithm is $A\left(X\left(\tau\right)\right)$. We treat $X\left(\infty\right)$
as the input to the function $f$ under the product distribution whose
mean is $x_{0}$, and say that the algorithm has 0 error if $A\left(X\left(\tau\right)\right)=f\left(X\left(\infty\right)\right)$
almost surely. The individual bit revealments are given by 
\begin{equation}
\delta_{i}:=\e\left[\left(X_{i}\left(\tau\right)-X_{i}\left(0\right)\right)^{2}\right]=\var\left(X_{i}\left(\tau\right)\right),\label{eq:revealment_of_axis_aligned_algorithm}
\end{equation}
and the total cost of algorithm is given by 
\begin{equation}
C\left(x_{0},Q\right):=\sum_{i=1}^{n}\delta_{i}.\label{eq:cost_of_axis_aligned_algorithm}
\end{equation}
We write $C\left(Q\right)=C\left(0,Q\right)$ for the cost of the
algorithm on the uniform measure. We denote the set of all $0$-error
axis-aligned algorithms by $\mathcal{S}$, and the set of all $0$-error
decision trees by $\mathcal{D}$. Decision trees are just axis-aligned
algorithms, where the jump size is $\eps=1$. 
\end{defn}

For an example of how the process $X\left(t\right)$ might look like,
see Figure \ref{fig:or_sample_paths}, which shows some sample paths
of the best axis-aligned algorithm for calculating the OR function
on two bits. 

\begin{figure}[H]
\includegraphics[scale=0.33]{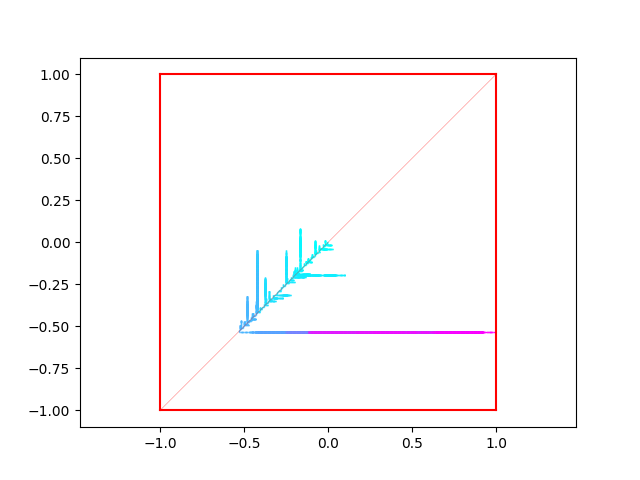}\includegraphics[scale=0.33]{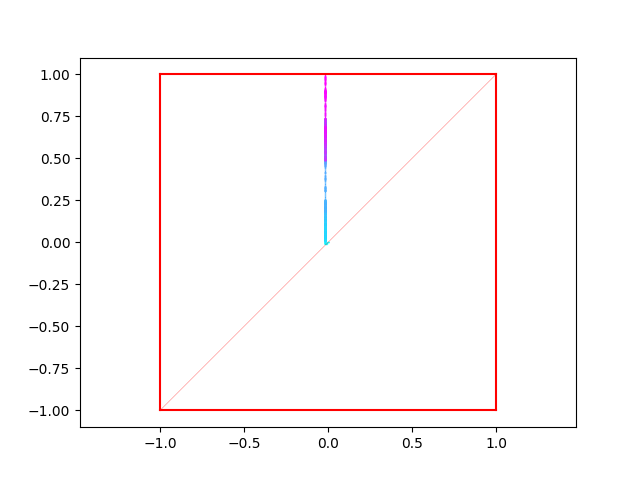}\includegraphics[scale=0.33]{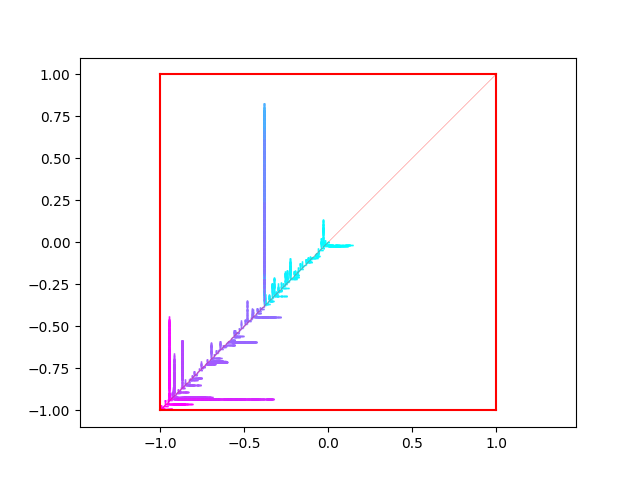}\caption{\label{fig:or_sample_paths}As will be shown in Section \ref{sec:or_example},
the fastest axis-aligned algorithm for computing the $n$-bit OR function
always updates the largest bit. Here are three sample paths for this
algorithm, with $n=2$ and $\protect\eps=2^{-7}$. The first two images
correspond to a function value of $1$, while the last image corresponds
to a value of $-1$.}
\end{figure}

\begin{rem}
The Fourier representation of $f:\left\{ -1,1\right\} ^{n}\to\r$
as a multilinear polynomial $f\left(x\right)=\sum_{S\subseteq\left[n\right]}\hat{f}\left(S\right)\prod_{i\in S}x_{i}$
allows us to extend $f$'s domain to the continuous cube $\left[-1,1\right]^{n}$.
This gives two natural choices for the output function $A$:
\begin{enumerate}
\item $A\left(x\right)=f\left(x\right)$. For any given stopping time $\tau$,
this choice minimizes the $\mathrm{L}^{2}$ error $\e\left[\left(f\left(X\left(\infty\right)\right)-A\left(X\left(\tau\right)\right)\right)^{2}\right]$
(to see this, differentiate the error by the Fourier coefficients
of $A$).
\item $A\left(x\right)=\sign\left(f\left(x\right)\right)$ (with arbitrary
$\pm1$ values when $f\left(x\right)=0$). When $f$ takes only the
values $\pm1$ on $\left\{ -1,1\right\} ^{n}$, for any given stopping
time $\tau$, this choice minimizes the error probability $\p\left[A\left(X\left(\tau\right)\right)\neq f\left(X\left(\infty\right)\right)\right]$.
\end{enumerate}
\end{rem}

By mimicking their coordinate-by-coordinate updates, axis-aligned
algorithms are a natural generalization of decision trees. However,
our revealment results apply to a larger class of algorithms, defined
below, which use more general processes to model their input. For
example, the process $X\left(t\right)$ can be driven by Brownian
motion, where only one coordinate can move at a given time. These
processes are common in optimal control theory (see e.g. \cite{mandelbaum_shepp_vanderbei_optimal_switching_between_brownian_motion}).
\begin{defn}[Fractional query algorithm]
\label{def:fractional_query_algorithm}Let $\mu$ be a product measure
with mean $x_{0}$ from which the inputs are drawn. A \emph{fractional
query algorithm} on inputs distributed by $\mu$ is a triplet $Q:=\left(X\left(t\right),\tau,A\right)$
such that:
\begin{enumerate}
\item $X\left(t\right)\in\left[-1,1\right]^{n}$ is a martingale, $X\left(0\right)=x_{0}$,
and $X\left(\infty\right)\sim\mu$. The time parameter can be either
discrete or continuous.
\item $\tau$ is an $X\left(t\right)$-adapted stopping time. 
\item $A$ is a function $A:\left[-1,1\right]^{n}\to\r$. 
\item \label{enu:linear_multivariates_are_martingales_condition}For every
$S\subseteq\left[n\right]$, the stochastic process $\left(\prod_{i\in S}X_{i}\left(t\right)\right)_{t}$
is a martingale. 
\end{enumerate}
The output of the algorithm is $A\left(X\left(\tau\right)\right)$.
The revealment and cost of fractional query algorithms are the same
as in equations (\ref{eq:revealment_of_axis_aligned_algorithm}) and
(\ref{eq:cost_of_axis_aligned_algorithm}). We denote the set of all
fractional query algorithms by $\mathcal{Q}$. 
\end{defn}

\begin{rem}
\label{rem:remark_on_quadratic_variation}Using the fact that the
expected value of the quadratic variation of a martingale is equal
to its variance, we can also write the cost as
\[
C\left(x_{0},Q\right)=\e\sum_{i=1}^{n}\left[X_{i}\right]_{\tau},
\]
where $\left[X_{i}\right]_{t}$ is the quadratic variation of $X_{i}$
at time $t$. When $X\left(t\right)$ is a discrete-time process,
this means that the cost of the algorithm is given by the expected
sum of squares of its jumps. 
\end{rem}

\begin{rem}
Condition (\ref{enu:linear_multivariates_are_martingales_condition})
ensures that the individual coordinates posses ``enough independence'',
and is needed for one of our theorems later on. It is satisfied for
all axis-aligned processes. 
\end{rem}

\subsubsection{Revealment bounds}

The class of fractional query algorithms $\mathcal{Q}$ contains the
class of decision trees $\mathcal{D}$, and so in general, for every
$i\in\left[n\right]$, $\inf_{Q\in\mathcal{Q}}\delta_{i}\leq\inf_{D\in\mathcal{D}}\delta_{i}$.
We show that several classical results concerning decision tree revealment
hold also for fractional query algorithms. We start with Theorem \ref{thm:ss_inequality}
by Schramm and Steif.
\begin{thm}
\label{thm:fractional_ss_10}Let the inputs be drawn from the uniform
distribution. Let $f:\left\{ -1,1\right\} ^{n}\to\r$, let $\left(X\left(t\right),\tau,f\right)$
be a fractional query algorithm for $f$ with $\mathrm{L}^{2}$ error
bounded by $\eps$:
\[
\e\left[\left(f\left(X\left(\infty\right)\right)-f\left(X\left(\tau\right)\right)\right)^{2}\right]\leq\eps^{2}.
\]
Denote $\delta=\max_{i\in\left[n\right]}\e\left[X_{i}\left(\tau\right)^{2}\right]$.
Then for every $k\in\left[n\right]$, 
\[
\sum_{\abs S=k}\hat{f}\left(S\right)^{2}\leq\left(\sqrt{\e\left(f\left(X\left(\tau\right)\right)\right)^{2}}\sqrt{k\delta}+\eps\right)^{2}.
\]
In particular, if $\left(X\left(t\right),\tau,f\right)$ has $0$
error, we recover Schramm and Steif's original bound:
\[
\sum_{\abs S=k}\hat{f}\left(S\right)^{2}\leq\delta k\norm f_{2}^{2}.
\]
\end{thm}

Next, using the computations in the proof of the above theorem, we
show that when restricted to bounded-degree functions, a weaker version
of the OSSS-inequality (\ref{eq:original_osss_inequality}) applies
for fractional query algorithms.
\begin{thm}
\label{thm:fractional_osss}Let the inputs be drawn from the uniform
distribution. Let $f:\left\{ -1,1\right\} ^{n}\to\r$ have degree
bounded by $d$, let $\left(X\left(t\right),\tau,f\right)$ be a $0$-error
fractional query algorithm for $f$, and let $\delta_{i}=\e\left[X_{i}\left(\tau\right)^{2}\right]$.
Then 
\begin{equation}
\var\left(f\right)\leq d\sum_{i}\delta_{i}\e\norm{\partial_{i}f}_{2}^{2}.\label{eq:slightly_stronger_osss_special_case}
\end{equation}
\end{thm}

This improves Corollary 3.4 in \cite{osss_every_decision_tree_has_an_influential_variable},
which had the same result for decision trees but with the leading
constant $d$ replaced by ``maximum depth of a decision tree'' (the
depth is always larger than the degree). Proving Theorem \ref{thm:osss_inequality}
in full generality remains an open problem; see Section \ref{subsec:open_questions}.

Finally, we remark that for $0$-error algorithms, the bound in Theorem
\ref{thm:bsw_inequality} holds also for fractional query algorithms. 
\begin{thm}
\label{thm:fractional_bsw}Let $f:\left\{ -1,1\right\} ^{n}\to\left\{ -1,1\right\} $
and let $\left(X\left(t\right),\tau,f\right)$ be a fractional query
algorithm which calculates $f$ exactly. Let $\delta=\max_{i}\e\left[X_{i}\left(\tau\right)^{2}\right]$.
Then 
\[
\delta\geq\sqrt{\frac{\var\left(f\right)}{2n}.}
\]
\end{thm}

In fact, this theorem stems from a similar result about function certificates
(see Section \ref{subsec:proof_of_bsw}), whose proof is basically
identical to that in \cite{bsw_balanced_boolean_functions}. This
gives hope that fractional query algorithm can achieve better revealment
bounds if they are allowed to make errors.

The proofs of the above theorems are found in Section \ref{sec:revealment_bounds}.

\subsubsection{Limits of axis-aligned algorithms}

When considering axis-aligned algorithms, it is natural to ask what
happens in the limit of $\eps\to0$. By Remark \ref{rem:remark_on_quadratic_variation},
for a fixed $\eps$, each step taken by the process increases the
cost by $\eps^{2}$, and so the task is to minimize the expected runtime.
Intuitively, we can extend the notion of axis-aligned processes to
start at any $x\in\left[-1,1\right]^{n}$, rather than just dyadic
starting positions. If we define $u_{\eps}\left(x\right):\left[-1,1\right]^{n}\to\r$
to be the cost of the best axis-aligned algorithm which starts at
$x$, then $u_{\eps}$ should satisfy 
\begin{equation}
u_{\eps}\left(x\right)=\min_{i}\frac{u_{\eps}\left(x+\eps e_{i}\right)+u_{\eps}\left(x-\eps e_{i}\right)}{2}+\eps^{2},\label{eq:dynamic_programming_in_intro}
\end{equation}
where $e_{i}$ is the unit vector in direction $i$. Taking $\eps\to0$,
the underlying processes $X\left(t\right)$ should converge to a continuous
time process, whose cost $u\left(x\right)$ at every starting point
$x$ is the best possible among all axis-aligned algorithms. Equation
(\ref{eq:dynamic_programming_in_intro}) suggests that $u$ should
satisfy $\min_{k}\frac{\partial^{2}}{\partial x_{k}^{2}}u=-2$ (see
Remark \ref{rem:intuition_behind_axis_aligned_laplacian}).

This intuition can be made precise in the framework of viscosity solutions
to partial differential equations (see Section \ref{subsec:background_viscosity}.
The main difficulty is that we have no assurance about the differentiability
of neither $u_{\eps}$ nor $u$). This is the main content of Section
\ref{sec:axis_aligned_processes_and_laplacian}, which culminates
in the following theorem.
\begin{thm}
\label{thm:intro_axis_aligned_laplacian}Let $f:\left\{ -1,1\right\} ^{n}\to\left\{ -1,1\right\} $
be non-constant, and let $u\left(x\right)=\lim_{\eps\to0}\inf_{Q_{\eps}}C\left(x,Q_{\eps}\right)$,
where $\eps$ is of the form $2^{-k}$ and the infimum is taken over
all axis-aligned algorithms with jump size $\eps$ starting at $x$.
Then $u$ is the unique continuous viscosity solution to the Dirichlet
boundary-value problem 
\[
\begin{cases}
\axis u=-2 & x\in\left(-1,1\right)^{n}\\
u=F & x\in\partial\left[-1,1\right]^{n},
\end{cases}
\]
where $\axis$ is the nonlinear operator $\min_{k}\frac{\partial^{2}}{\partial x_{k}^{2}}$,
and $F$ is obtained by recursively solving the Dirichlet problem
on the $\left(n-1\right)$ dimensional facets for the appropriate
restrictions of $f$. 
\end{thm}

As a possible application of this theorem, consider the class of transitive
Boolean functions (see Section \ref{subsec:background_boolean} for
a definition). For such functions, the individual bit revealments
can be made equal under the uniform distribution, and a query algorithm
$Q$ yields $\delta=C\left(Q\right)/n$. Theorem \ref{thm:fractional_ss_10}
then gives 
\[
\sum_{\abs S=k}\hat{f}\left(S\right)^{2}\leq\frac{u\left(0\right)}{n}k\norm f_{2}^{2},
\]
allowing us to obtain noise sensitivity bounds by solving PDEs. 

Using this framework, in Section \ref{sec:or_example} we calculate
$u\left(x\right)$ exactly for the OR function between two bits, and
show that as $n\to\infty$, fractional query algorithms need to query
only $1$ bit in expectation in order to calculate the OR function
on $n$ bits (whereas decision trees need to query $2$ bits in expectation).
This is the best separation we have found so far between decision
trees and fractional query algorithms; we do not know whether an asymptotic
separation is possible (see Section \ref{subsec:open_questions}).
We also propose a heuristic, which might prove useful in analyzing
recursive functions. 

\subsubsection{Fractional random-turn games}

Finally, in Section \ref{sec:fractional_random_turn_games} we show
that random-turn games, introduced by Peres, Schramm, Sheffield and
Wilson in \cite{peres_schramm_sheffield_wilson_random_turn_hex},
yield natural axis-aligned algorithms. However, these algorithms do
not necessarily minimize the cost $C\left(Q\right)$. 

\subsection{Open questions\label{subsec:open_questions}}
\begin{enumerate}
\item We have so far been able to neither prove nor disprove that fractional
query algorithms can run asymptotically faster than decision trees.
Does there exist a global constant $M$, so that for every function
$f:\left\{ -1,1\right\} ^{n}\to\left\{ -1,1\right\} $, $\inf_{Q\in\mathcal{Q}}C\left(Q\right)\geq M\inf_{D\in\mathcal{D}}\left(D\right)$?
We find both possible answers to be exciting: if fractional query
algorithms can run asymptotically faster (say, on transitive monotone
functions), then Theorem \ref{thm:fractional_ss_10} gives improved
noise sensitivity estimates; and if fractional query algorithms are
equivalent to decision trees, then Theorem \ref{thm:intro_axis_aligned_laplacian}
allows us to obtain decision tree lower bounds by solving partial
differential equations.
\item Show that for functions $f:\left\{ -1,1\right\} ^{n}\to\left\{ -1,1\right\} $,
Theorem \ref{thm:osss_inequality} applies to fractional query algorithms.
The problem, as is alluded to in \cite{osss_every_decision_tree_has_an_influential_variable},
is that the inequality (\ref{eq:original_osss_inequality}) is at
heart an $\mathrm{L}^{1}$ bound, while the fractional query costs
are in essence $\mathrm{L}^{2}$ variables: since $X\left(\tau\right)$
is a martingale, the cost $\e X_{i}\left(\tau\right)^{2}$ incurred
by bit $i$ is equal to the expected quadratic variation $\e\left[X_{i}\right]_{\tau}$.
It is possible to define instead the cost by the total variation of
$X_{i}\left(\tau\right)$, and get an analogous inequality for this
alternate cost, but this would undesirably rule out a large class
of processes $X\left(t\right)$ which have infinite total variation
(such as Brownian motion). 
\item Can Theorem \ref{thm:bsw_inequality} be extended to fractional query
algorithms which are allowed to make errors?
\item Let $\mathcal{S}$ be the set of axis-aligned algorithms. Show that
these are the best algorithms possible, i.e. that $\inf_{Q\in\mathcal{Q}}C\left(Q\right)=\inf_{S\in\mathcal{S}}C\left(S\right)$.
\item Show that when the boundary conditions are recursively defined by
Boolean functions as in Theorem \ref{thm:intro_axis_aligned_laplacian},
the limiting value function $u$ of axis-aligned algorithms is twice
continuously differentiable (thereby eliminating the need to use viscosity
solutions, and making life easier for the working mathematician).
\item Show that for axis-aligned algorithms with jump size $\eps\to0$,
there is a sense in which the underlying process $X\left(t\right)$
converges to Brownian motion, where only one coordinate can move at
a time (see e.g. \cite{jacka_warren_windridge_minimizing_time_to_a_decision,mandelbaum_shepp_vanderbei_optimal_switching_between_brownian_motion}
for an exact formulation of this framework). 
\end{enumerate}

\subsection{Related fractional work}

In the context of axis-aligned algorithms, Jacka, Warren and Windridge
\cite{jacka_warren_windridge_minimizing_time_to_a_decision} inspect
the minimum time for continuous processes to compute the value of
$3$-majority. While this can be seen as a problem in optimal control
theory, most work in the literature focus on different payoff models.
See Section 1.1 in \cite{jacka_warren_windridge_minimizing_time_to_a_decision}
for more detail. 

Fractional processes appear in other areas of Boolean analysis. In
\cite{chattopadahyay_hatami_hosseini_lovett_polarizing}, Chattopadhyay,
Hatami, Hosseini and Lovett construct a random walk $X\left(t\right)$
so that 
\begin{equation}
\abs{\e\left[f\left(X\left(\tau\right)\right)-\e f\right]}\leq\eps\label{eq:polarizing_prgs}
\end{equation}
for functions with bounded Fourier tails. In some sense, they approach
the problem from a different viewpoint than us: in order for (\ref{eq:polarizing_prgs})
to be meaningful, they require that $\e\left[X_{i}\left(\tau\right)^{2}\right]$
is large, whereas our theorems are useful when $\e\left[X_{i}\left(\tau\right)^{2}\right]$
is small. This stresses the difference between noise sensitive functions
and functions with bounded Fourier tails. 

Axis-aligned processes can be seen as a type of single player ``tug-of-war''
game. The influential work by Peres, Schramm, Sheffield and Wilson
\cite{peres_schramm_sheffield_wilson_tug_of_war} considers a two-player
version. There, a dynamic programming equation also leads to a PDE,
involving the infinity Laplacian $\Delta_{\infty}$ rather than the
axis-aligned Laplacian $\axis$.

Algorithms which make fractional queries have already been investigated
in the quantum setting, by allowing gates to impart fractional phases
(or, equivalently, by running a Hamiltonian for a small amount of
time). It has been shown that if errors are allowed, quantum fractional
query algorithms are not much more powerful than quantum discrete
query algorithms. See e.g. \cite{simulating_sparse_hamiltonians_quantum,kothari_thesis}.

The process $X\left(t\right)$ can be seen as a noisy version of the
true input $X\left(\infty\right)$, where more confidence is given
to coordinates with larger absolute value. A different model based
on noisy inputs is given by Ben-David and Blais in \cite{ben_david_blais_noisy_oracle}.
They consider decision trees which can make noisy queries, and pay
a cost of $\gamma^{2}$ in order to get a random bit $\tilde{x}_{i}\in\left\{ -1,1\right\} $
that is $\gamma$-correlated with the true input $x_{i}$. 

\subsection{Acknowledgments}

We thank Ronen Eldan and Shachar Lovett for comments and discussions. 

\section{Background and notation}

We denote the standard basis of $\r^{n}$ by $\left\{ e_{1},\ldots,e_{n}\right\} $.
For a set $\Omega\subseteq\r^{n}$, we denote its boundary by $\partial\Omega$.

\subsection{Boolean functions\label{subsec:background_boolean}}

For a general introduction to Boolean functions, see \cite{odonnell_analysis_of_boolean_functions};
in what follows, we provide a brief overview of the required background
and notation.

A distribution $\mu$ over $\left\{ -1,1\right\} ^{n}$ is a product
measure if, for $x\sim\mu$, the bits $x_{i}$ are all independent.
Denote by $\mu_{1/2}$ the uniform measure on $\left\{ -1,1\right\} ^{n}$.
For every function $f:\left\{ -1,1\right\} ^{n}\to\r$, its expectation
and variance are given by 
\[
\e f=\e_{y\sim\mu_{1/2}}f\left(y\right)\text{ and }\var\left(f\right)=\e\left(f-\e f\right)^{2}.
\]
Every function $f:\left\{ -1,1\right\} ^{n}\to\r$ may be uniquely
written as a sum of monomials:
\begin{equation}
f\left(y\right)=\sum_{S\subseteq\left[n\right]}\hat{f}\left(S\right)\prod_{i\in S}y_{i},\label{eq:definition_of_fourier_decomposition}
\end{equation}
where $\left[n\right]=\left\{ 1,\ldots,n\right\} $, and $\hat{f}\left(S\right)$
are known as the Fourier coefficients. By Parseval's identity, the
squared norm of a function is given by 
\[
\norm f_{2}^{2}:=\e\left[f^{2}\right]=\sum_{S\subseteq\left[n\right]}\hat{f}\left(S\right)^{2},
\]
while the variance of a function is given by the sum of its Fourier
weights of level greater than $0$: 
\[
\var\left(f\right)=\sum_{S\neq\emptyset}\hat{f}\left(S\right)^{2}.
\]
Equation (\ref{eq:definition_of_fourier_decomposition}) may be used
to extend a function's domain from the discrete hypercube $\left\{ -1,1\right\} ^{n}$
to real space $\r^{n}$. We call this the \emph{harmonic extension},
and denote it also by $f$. Under this notation, $f\left(0\right)=\e f$.
The derivative of a function $f$ in direction $i$ is defined as
\[
\partial_{i}f\left(y\right)=\frac{f\left(y^{i\to1}\right)-f\left(y^{i\to-1}\right)}{2},
\]
where $y^{i\to a}$ has $a$ at coordinate $i$, and is identical
to $y$ at all other coordinates. The Fourier representation of the
derivative is given by 
\begin{equation}
\partial_{i}f=\sum_{S\ni i}\hat{f}\left(S\right)\prod_{j\in S,j\neq i}x_{j}.\label{eq:fourier_representation_of_derivative}
\end{equation}
A function is called \emph{monotone} if $f\left(x\right)\leq f\left(y\right)$
whenever $x_{i}\leq y_{i}$ for all $i\in\left[n\right]$. A function
is called \emph{transitive} if for all $i,j\in\left[n\right]$, there
exists a permutation $\pi$ on the $n$ bits with $\pi\left(i\right)=j$
such that for all $x\in\left\{ -1,1\right\} ^{n}$, $f\left(x\right)=f\left(\pi\left(x\right)\right)$. 

\subsection{Viscosity solutions\label{subsec:background_viscosity}}
\begin{defn}
For a twice-differentiable function $u:\r^{n}\to\r$, define the \emph{axis-aligned
Laplacian} $\axis$ by 
\[
\axis u=\min_{k}\frac{\partial^{2}u}{\partial x_{k}^{2}}.
\]
In this work, we will need to apply the operator $\axis$ to functions
$u$ which are continuous, but might not necessarily be twice- or
even once-differentiable. One natural framework for this is that of
\emph{viscosity solutions}. For a good introduction, see \cite{crandall_ishii_lions_users_guide_to_viscosity}.
We will require only the basic definitions.
\end{defn}

\begin{defn}
\label{def:viscosity_solutions}Let $g:\r^{n}\to\r$. The function
$u:\r^{n}\to\r$ is called a 
\begin{enumerate}
\item \emph{\label{enu:viscosity_subsolution_definition}viscosity subsolution}
to the equation $\axis u=g$, if for every smooth function $\varphi:\r^{n}\to\r$
such that $\varphi-u$ has a minimum at point $x_{0}$, we have $\axis\varphi\left(x_{0}\right)\geq g\left(x_{0}\right)$;
\item \emph{viscosity supersolution }to the equation $\axis u=g$,\emph{
}if for every smooth function $\varphi:\r^{n}\to\r$ such that $\varphi-u$
has a maximum at point $x_{0}$, we have $\axis\varphi\left(x_{0}\right)\leq g\left(x_{0}\right)$; 
\item \emph{viscosity solution} to the equation $\axis u=g$ if it is both
a viscosity subsolution and a viscosity supersolution.
\end{enumerate}
\end{defn}

\begin{rem}
Suppose $\varphi\left(x\right)$ is such that $\varphi-u$ has a minimum
at $x_{0}$ but $\axis\varphi\left(x_{0}\right)<g\left(x_{0}\right)$.
Then for $\eps>0$ small enough, the function $\varphi'\left(x\right):=\varphi\left(x\right)+\eps\norm{x-x_{0}}_{2}^{2}$
also satisfies $\axis\varphi'\left(x_{0}\right)<g\left(x_{0}\right)$,
and has a strict minimum at $x_{0}$. When checking whether $u$ is
a viscosity solution to $\axis u=g$, we can therefore restrict ourselves
to test functions $\varphi$ with strict minima / maxima at $x_{0}$. 
\end{rem}

\section{Classical revealment bounds\label{sec:revealment_bounds}}

\subsection{Theorem \ref{thm:fractional_ss_10}}

The proof follows the lines of the original proof by Schramm and Steif.
Some additional calculations must be made when using the language
of fractional query algorithms. 
\begin{proof}[Proof of Theorem \ref{thm:fractional_ss_10}]
Let $g:\left\{ -1,1\right\} ^{n}\to\r$ be defined as $g\left(x\right)=\sum_{\abs S=k}\hat{f}\left(S\right)\prod_{i\in S}x_{i}$.
Let $Z\left(t\right)=f\left(X\left(\infty\right)\right)-f\left(X\left(t\right)\right)$
be the error of the algorithm as a function of time. Since $X\left(\infty\right)$
is uniform on $\left\{ -1,1\right\} ^{n}$, by Plancharel's inequality,
we have, for all $t$,

\[
\norm g_{2}^{2}\overset{\text{Plancharel}}{=}\e\left[g\left(X\left(\infty\right)\right)\left(Z\left(t\right)+f\left(X\left(t\right)\right)\right)\right].
\]
By the Cauchy-Schwarz inequality, for $t=\tau$, 
\begin{align}
\e\left[g\left(X\left(\infty\right)\right)\left(Z\left(\tau\right)\right)\right] & \leq\sqrt{\e g\left(X\left(\infty\right)\right)^{2}}\sqrt{\e Z\left(\tau\right)^{2}}\nonumber \\
 & \leq\norm g_{2}\eps.\label{eq:error_in_SS10}
\end{align}
Since $f\left(X\left(\tau\right)\right)$ is determined by $X\left(\tau\right)$,
we have 
\begin{align}
\e\left[g\left(X\left(\infty\right)\right)f\left(X\left(\tau\right)\right)\right] & =\e\left[\e\left[g\left(X\left(\infty\right)\right)f\left(X\left(\tau\right)\right)\mid X\left(\tau\right)\right]\right]\nonumber \\
 & =\e\left[f\left(X\left(\tau\right)\right)\e\left[g\left(X\left(\infty\right)\right)\mid X\left(\tau\right)\right]\right]\nonumber \\
 & \overset{\text{Property \eqref{enu:linear_multivariates_are_martingales_condition}}}{=}\e\left[f\left(X\left(\tau\right)\right)g\left(X\left(\tau\right)\right)\right]\nonumber \\
 & \overset{\text{Cauchy-Schwarz}}{\leq}\sqrt{\e\left[f\left(X\left(\tau\right)\right)^{2}\right]}\sqrt{\e\left[g\left(X\left(\tau\right)\right)^{2}\right]}.\label{eq:SS10_starting_point}
\end{align}
To bound the term $\e\left[g\left(X\left(\tau\right)\right)^{2}\right]$,
we show how the harmonic extension is related to an interpolation
function. Let $x\in\left[-1,1\right]^{n}$. For any function $h:\left\{ -1,1\right\} ^{n}\to\r$
with Fourier decomposition $h\left(y\right)=\sum_{S}\hat{h}\left(S\right)\prod_{i\in S}y_{i}$,
let $h_{x}:\left\{ -1,1\right\} ^{n}\to\r$ be defined by mapping
$y_{i}\mapsto\sqrt{1-x_{i}^{2}}y_{i}+x_{i}$ for every $i$: 
\begin{equation}
h_{x}\left(y\right)=\sum_{S}\hat{h}\left(S\right)\prod_{i\in S}\left(\sqrt{1-x_{i}^{2}}y_{i}+x_{i}\right).\label{eq:interpolating_function}
\end{equation}
Note that the expected value of $h_{x}$ is given by $\e h_{x}=\widehat{h_{x}}\left(\emptyset\right)=\sum_{S}\hat{h}\left(S\right)\prod_{i\in S}x_{i}$.
This is exactly the same as inputting the vector $x$ into the harmonic
extension of $h$. Thus 
\[
h\left(x\right)=\widehat{h_{x}}\left(\emptyset\right).
\]
Taking $h=g$ and $x=X\left(\tau\right)$ in (\ref{eq:interpolating_function}),
we have 
\begin{align*}
\e\left[g\left(X\left(\tau\right)\right)^{2}\right] & =\e\left[\widehat{g_{X\left(\tau\right)}}\left(\emptyset\right)^{2}\right]\\
 & \overset{\text{Parseval}}{=}\e\left[\norm{g_{X\left(\tau\right)}}_{2}^{2}-\sum_{\abs S>0}\widehat{g_{X\left(\tau\right)}}\left(S\right)^{2}\right].
\end{align*}
By definition of the $2$-norm, 
\begin{align*}
\e\norm{g_{X\left(\tau\right)}}_{2}^{2} & =\e_{X\left(\tau\right)}\e_{y}\left[g_{X\left(\tau\right)}\left(y\right)^{2}\right]\\
 & =\e_{X\left(\tau\right)}\e_{y}\sum_{S,T\subseteq\left[n\right]}\hat{g}\left(S\right)\hat{g}\left(T\right)\prod_{i\in S}\left(\sqrt{1-X_{i}\left(\tau\right)^{2}}y_{i}+X_{i}\left(\tau\right)\right)\prod_{j\in T}\left(\sqrt{1-X_{j}\left(\tau\right)^{2}}y_{j}+X_{j}\left(\tau\right)\right).
\end{align*}
If $i\in S\intersect T$, then 
\[
\e_{y}\left[\left(\sqrt{1-X_{i}\left(\tau\right)^{2}}y_{i}+X_{i}\left(\tau\right)\right)^{2}\right]=\e_{y}\left[\left(1-X_{i}\left(\tau\right)^{2}\right)y_{i}^{2}+X_{i}\left(\tau\right)^{2}-2y_{i}X_{i}\left(\tau\right)\sqrt{1-X_{i}\left(\tau\right)^{2}}\right]\overset{y_{i}^{2}=1,\e y_{i}=0}{=}1.
\]
By independence of the coordinates of $y$, indices in $S\intersect T$
therefore do not contribute to the product $\prod_{i\in S}\left(\ldots\right)\prod_{j\in T}\left(\ldots\right)$,
and we are left with a product of the form $\prod_{i\in S\backslash T}\left(\ldots\right)\prod_{j\in T\backslash S}$$\left(\ldots\right)$.
Since $\e y_{i}=0$ for all $i\in\left[n\right]$ and $\e\prod_{i\in A}X_{i}\left(\tau\right)=0$
for all $A\subseteq\left[n\right]$, the coefficient of $\hat{g}\left(S\right)\hat{g}\left(T\right)$
is $0$ if $S\neq T$, and $1$ if $S=T$. Thus 
\[
\e\norm{g_{X\left(\tau\right)}}_{2}^{2}=\sum_{S\subseteq\left[n\right]}\hat{g}\left(S\right)^{2},
\]
and we get 
\begin{align}
\e\left[g\left(X\left(\tau\right)\right)^{2}\right] & =\sum_{S\subseteq\left[n\right]}\hat{g}\left(S\right)^{2}-\e\sum_{\abs S>0}\widehat{g_{X\left(\tau\right)}}\left(S\right)^{2}\nonumber \\
 & \overset{\text{by definition of g}}{=}\sum_{\abs S=k}\hat{g}\left(S\right)^{2}-\e\sum_{\abs S>0}\widehat{g_{X\left(\tau\right)}}\left(S\right)^{2}\nonumber \\
 & \leq\sum_{\abs S=k}\left(\hat{g}\left(S\right)^{2}-\e\widehat{g_{X\left(\tau\right)}}\left(S\right)^{2}\right).\label{eq:middle_bound_on_g}
\end{align}
Since $g$ itself does not have Fourier coefficients of frequencies
larger than $k$, the only sets $S$ with $\abs S=k$ in the Fourier
decomposition of $g_{X\left(\tau\right)}$ can come from choosing
the factor $\sqrt{1-X_{i}\left(\tau\right)^{2}}$ every time in the
product in (\ref{eq:interpolating_function}). Thus, for $\abs S=k$,
\begin{align}
\hat{g}\left(S\right)^{2}-\widehat{g_{X\left(\tau\right)}}\left(S\right)^{2} & =\hat{g}\left(S\right)^{2}\left(1-\prod_{i\in S}\left(\sqrt{1-X_{i}\left(\tau\right){}^{2}}\right)^{2}\right)\nonumber \\
 & =\hat{g}\left(S\right)^{2}\left(1-\prod_{i\in S}\left(1-X_{i}\left(\tau\right){}^{2}\right)\right)\nonumber \\
 & \overset{\text{Weierstrass inequality}}{\leq}\hat{g}\left(S\right)^{2}\left(1-\left(1-\sum_{i\in S}X_{i}\left(\tau\right){}^{2}\right)\right)\nonumber \\
 & \leq\hat{g}\left(S\right)^{2}\sum_{i\in S}X_{i}\left(\tau\right)^{2}.\label{eq:late_bound_on_g}
\end{align}
Plugging this back into (\ref{eq:middle_bound_on_g}), we get
\begin{align*}
\e\left[g\left(X\left(\tau\right)\right)^{2}\right] & \leq\sum_{\abs S=k}\hat{g}\left(S\right)^{2}\sum_{i\in S}\e X_{i}\left(\tau\right)^{2}\\
 & \leq\sum_{\abs S=k}\hat{g}\left(S\right)^{2}k\delta=\norm g_{2}^{2}k\delta.
\end{align*}
Together with (\ref{eq:error_in_SS10}) and (\ref{eq:SS10_starting_point}),
this yields the desired result.
\end{proof}

\subsection{Theorem \ref{thm:fractional_osss}}
\begin{proof}[Proof of Theorem \ref{thm:fractional_osss}]
Since both the left-hand side and the right-hand side of (\ref{eq:slightly_stronger_osss_special_case})
are invariant to shifts of the type $f\mapsto f+c$, we can assume
without loss of generality that $\e f=0$. 

Let $k\in\n$. As in the proof of Theorem \ref{thm:fractional_ss_10},
define $g_{k}:\left\{ -1,1\right\} ^{n}\to\r$ by 
\[
g_{k}\left(x\right)=\sum_{\abs S=k}\hat{f}\left(S\right)\prod_{i\in S}x_{i}.
\]
By equations (\ref{eq:SS10_starting_point}), (\ref{eq:middle_bound_on_g})
and (\ref{eq:late_bound_on_g}) in the proof of Theorem \ref{thm:fractional_ss_10},
using the fact that $f$ is computed exactly, we have 
\begin{align*}
\left(\norm{g_{k}}_{2}^{2}\right)^{2} & \leq\norm f_{2}^{2}\sum_{\abs S=k}\widehat{g_{k}}\left(S\right)^{2}\sum_{i\in S}\e X_{i}\left(\tau\right)^{2}\\
 & =\norm f_{2}^{2}\sum_{i=1}^{n}\delta_{i}\sum_{S\ni i,\abs S=k}\hat{f}\left(S\right)^{2}.
\end{align*}
Supposing that $f$ is a degree-$d$ polynomial, we can sum the above
for all $k=1,\ldots,d$. Recalling that $\partial_{i}f=\sum_{S\ni i}\hat{f}\left(S\right)\chi_{S\backslash\left\{ i\right\} }$(see
(\ref{eq:fourier_representation_of_derivative})), we get $\norm{\partial_{i}f}_{2}^{2}=\sum_{S\ni i}\hat{f}\left(S\right)^{2}$,
and so
\begin{equation}
\sum_{k=1}^{d}\left(\norm{g_{k}}_{2}^{2}\right)^{2}\leq\norm f_{2}^{2}\sum_{i=1}^{n}\delta_{i}\norm{\partial_{i}f}_{2}^{2}.\label{eq:osss_almost_there}
\end{equation}
The left hand side gives an upper bound to $\norm f_{2}^{4}$: 
\[
\left(\norm f_{2}^{2}\right)^{2}=\left(\sum_{k=1}^{d}\norm{g_{k}}_{2}^{2}\right)^{2}\overset{\text{Cauchy-Schwarz}}{\leq}d\sum_{k=1}^{d}\left(\norm{g_{k}}_{2}^{2}\right)^{2},
\]
and plugging this into (\ref{eq:osss_almost_there}) gives
\[
\norm f_{2}^{2}\leq d\sum_{i=1}^{n}\delta_{i}\norm{\partial_{i}f}_{2}^{2}.
\]
Since we assume that $\e f=0$, we have $\var\left(f\right)=\norm f_{2}^{2}$,
and the result follows. 
\end{proof}

\subsection{Theorem \ref{thm:fractional_bsw}\label{subsec:proof_of_bsw}}

The proof does not use the fact that the processes are allowed to
take fractional values in $\left[-1,1\right]^{n}$, and instead uses
a result about certificates. For $x\in\left\{ -1,1\right\} ^{n}$,
an $f\left(x\right)$-certificate is a minimal set $S\subseteq\left[n\right]$
such that for all $y\in\left\{ -1,1\right\} ^{n}$ with $x_{i}=y_{i}$
for $i\in S$, we have $f\left(x\right)=f\left(y\right)$. In other
words, the value of $f$ does not change if we change the bits outside
of $S$. Theorem \ref{thm:fractional_bsw} is actually a special case
of the following average-case certificate complexity result, whose
proof follows that of Theorem \ref{thm:bsw_inequality} in \cite{bsw_balanced_boolean_functions}.
\begin{lem}
For every $x\in\left\{ -1,1\right\} ^{n}$, let $\Gamma_{x}$ be a
probability distribution on $2^{\left[n\right]}$ supported on the
set of $f\left(x\right)$-certificates, and let $S_{x}\sim\Gamma_{x}$.
Let $\delta_{i}=\p_{x\sim\mathrm{Unif}\left\{ -1,1\right\} ^{n}}\left[i\in S_{x}\right]$
and $\delta=\max_{i}\delta_{i}$. Then 
\[
\delta\geq\sqrt{\frac{\var\left(f\right)}{2n}}.
\]
\end{lem}

\begin{proof}
Let $x,y\in\left\{ -1,1\right\} ^{n}$ be independent and uniformly
random. Let $N=\abs{S_{x}\intersect S_{y}}$. Then 
\[
\p\left[N>0\right]\leq\e\left[N\right]=\sum_{i=1}^{n}\delta_{i}^{2}\leq n\delta^{2}.
\]
Let $z\in\left\{ -1,1\right\} ^{n}$ be defined by 
\[
z_{i}=\begin{cases}
x_{i} & i\in S_{x}\\
y_{i} & \text{otherwise.}
\end{cases}
\]
Clearly, $f\left(x\right)=f\left(z\right)$, since both $x$ and $z$
are identical on the same $f\left(x\right)$-certificate $S_{x}$.
If $N=0$, then $z$ is also identical with $y$ on $S_{y}$, and
so $f\left(y\right)=f\left(z\right)$, implying that $f\left(x\right)=f\left(y\right)$.
We then have
\begin{align*}
1 & =\p\left[N=0\right]+\p\left[N>0\right]\\
 & \leq\p\left[f\left(x\right)=f\left(y\right)\right]+n\delta^{2}\\
 & =\p\left[f\left(x\right)=1\right]^{2}+\left(1-\p\left[f\left(x\right)=1\right]\right)^{2}+n\delta^{2}\\
 & =1-\frac{1}{2}\var\left(f\right)+n\delta^{2},
\end{align*}
yielding the result. 
\end{proof}
\begin{proof}[Proof of Theorem \ref{thm:fractional_bsw}]
If a fractional query algorithm calculates $f$ exactly, it must
completely reveal all the bits of some $f\left(X\left(\infty\right)\right)$
certificate. Hence, for every $x\in\left\{ -1,1\right\} ^{n}$, every
fractional query algorithm induces a probability distribution $\Gamma_{x}$
on $f\left(x\right)$ certificates, and $\e\left[X_{i}\left(\tau\right)^{2}\right]\geq\p_{x\sim\mathrm{Unif}\left\{ -1,1\right\} ^{n}}\left[i\in S_{x}\right]$. 
\end{proof}

\section{\label{sec:axis_aligned_processes_and_laplacian}Dynamic programming
and the axis-aligned Laplacian}

In this section, we show that the cost $u_{\eps}$ of axis-aligned
query algorithms converges as $\eps\to0$ (Proposition \ref{prop:limiting_cost_of_axis_aligned}).
We show that that $u$ is Lipschitz, and so the convergence is uniform
(Lemma \ref{lem:u_eps_are_lipschitz} and Lemma \ref{lem:uniform_convergence}).
This allows us to prove that $\axis u=-2$ (Theorem \ref{thm:obeying_axis_aligned_laplacian}).
Uniqueness follows from a general comparison principle on the axis-aligned
Laplacian (Theorem \ref{thm:comparison_principle}).

We first redefine axis-aligned processes to allow starting at non-dyadic
points $x\in\left[-1,1\right]^{n}$. In order to keep the process
inside $\left[-1,1\right]^{n}$ at all times, and still preserve its
martingale nature, this forces some changes to the update rule near
the boundary of the cube.
\begin{defn}[Axis-aligned jump process]
\label{def:general_axis_aligned_process}Let $\eps>0$. An \emph{axis-aligned
jump process} with jump size $\eps$ is a discrete-time martingale
$X\left(t\right)\in\left[-1,1\right]^{n}$. At time $t$, if $X\left(t\right)\notin\left\{ -1,1\right\} ^{n}$,
a direction $i_{t}\in\left[n\right]$ with $X_{i_{t}}\left(t\right)\in\left(-1,1\right)$
is chosen according to a\emph{ direction choosing strategy} $S$.
Let $a_{t},b_{t}$ be defined as follows.
\begin{enumerate}
\item If the distance between $X_{i_{t}}\left(t\right)$ to both endpoints
$\left\{ -1,1\right\} $ is at least $\eps$, i.e. $X_{i_{t}}\left(t\right)+\eps<1$
and $X_{i_{t}}\left(t\right)-\eps>-1$, then $a_{t}=X_{i_{t}}\left(t\right)-\eps$
and $b_{t}=X_{i_{t}}\left(t\right)+\eps$.
\item If $X_{i_{t}}\left(t\right)+\eps>1$, then $a_{t}=1-2\eps$ and $b_{t}=1$.
\item If $X_{i_{t}}\left(t\right)-\eps<-1$, then $a_{t}=-1$ and $b_{t}=-1+2\eps$. 
\end{enumerate}
The new position $X\left(t+1\right)$ is then updated in a martingale
fashion: for all $j\neq i_{t}$, we have $X_{j}\left(t+1\right)=X_{j}\left(t\right)$,
and for $i_{t}$ we have 
\[
X_{i_{t}}\left(t+1\right)=\begin{cases}
a_{t} & \text{with probability }\frac{b_{t}-X_{i_{t}}\left(t\right)}{2\eps}\\
b_{t} & \text{with probability }\frac{X_{i_{t}}\left(t\right)-a_{t}}{2\eps}.
\end{cases}
\]

\end{defn}

We will deal only with $0$-error axis-aligned query algorithms. We
therefore always set $\tau=\inf\left\{ t>0\mid f\left(X\left(t\right)\right)\in\left\{ -1,1\right\} \right\} $
and $A\left(x\right)=f\left(x\right)$. Each direction choosing strategy
$S$ thus determines an axis-aligned query algorithm. 

For a given strategy $S$, denote its expected cost when starting
at $X\left(0\right)=x$ by $C\left(x,S\right)=\sum_{i=1}^{n}\var\left(X_{i}\left(\tau\right)\right)$.
For $x\in\left[-1,1\right]^{n}$, let $u_{\eps}:\left[-1,1\right]^{n}\to\r$
be the best possible cost over all axis-aligned algorithms:
\[
u_{\eps}\left(x\right)=\inf_{S}C\left(x,S\right).
\]
Under strategy $S$, the cost when starting at $x\in\left[-1,1\right]^{n}$
can be broken up into two parts: the cost for taking a single step,
given by $\e\norm{X\left(1\right)-x}_{2}^{2}$, and the cost of continuing
the strategy starting at $X\left(1\right)$, given by $\e C\left(X\left(1\right),S'\right)$,
where $S'$ is some other strategy related to $S$, which incorporates
the knowledge that $t=1$ and $X\left(0\right)=x$. Thus
\[
C\left(x,S\right)=\e\left[\norm{X\left(1\right)-X\left(0\right)}_{2}^{2}+C\left(X\left(1\right),S'\right)\right].
\]
The right-hand side is a convex combination of choices of directions,
and the strategy $S$ is always improved by picking the minimal direction.
It is then clear that in order to minimize the cost, it is enough
to consider only deterministic Markov strategies. Since each Markov
strategy only has finitely many directions to choose from, there must
in fact exist an optimal strategy.
\begin{fact}
There is a Markov strategy $S$ such that for every $x\in\left[-1,1\right]^{n}$,
$u_{\eps}\left(x\right)=C\left(x,S\right)$. For $x\in\left[-1+\eps,1-\eps\right]^{n}$,
we have 
\begin{equation}
u_{\eps}\left(x\right)=\min_{i}\frac{u_{\eps}\left(x+\eps e_{i}\right)+u_{\eps}\left(x+\eps e_{i}\right)}{2}+\eps^{2}.\label{eq:dynamic_programming_equation}
\end{equation}
\end{fact}

We can construct an equivalent, recursive formulation of the cost
function $u_{\eps}$. In this variant, we are given a function $F:\partial\left[-1,1\right]^{n}\to\r$
defined on the boundary of the hypercube. We again run a process $X\left(t\right)$
while choosing directions using a strategy $S$, but the goal is now
to minimize the running-cost plus the exit-cost $F$. More formally,
let $\sigma=\inf\left\{ t>0\mid X\left(t\right)\in\partial\left[-1,1\right]^{n}\right\} $
be the first time that the process $X\left(t\right)$ hits the boundary
of the hypercube. The cost of the strategy $S$ in this model is 
\[
C_{F}\left(x,S\right)=\sum_{i=1}^{n}\var\left(X_{i}\left(\sigma\right)\right)+\e F\left(X\left(\sigma\right)\right).
\]
When $F$ is defined recursively on the facets of the hypercube as
the best cost among all axis-aligned algorithms, then minimizing $C_{F}\left(x,S\right)$
is equivalent to minimizing $C\left(x,S\right)$ (i.e. start with
$F\left(x\right)=0$ for all $x\in\left\{ -1,1\right\} ^{n}$, then
iteratively set the value of $F$ for all lines, squares, cubes, etc.
by considering axis-aligned algorithms with jump size $\eps$). In
this case, the cost $u_{\eps}$ equals $F$ on the boundary. This
formulation is useful for proving the following lemma.
\begin{lem}
\label{lem:u_eps_are_lipschitz}Let $n>0$ be an integer and let $L_{n}=\frac{n+3}{2}$.
Then $u_{\eps}$ is $L_{n}$-Lipschitz under the one-norm, i.e. $\abs{u_{\eps}\left(x\right)-u_{\eps}\left(y\right)}\leq\frac{n+3}{2}\norm{x-y}_{1}$
for all $x,y\in\left[-1,1\right]^{n}$. 
\end{lem}

\begin{proof}
By induction on the dimension. For $n=1$, if $f$ is constant then
$u_{\eps}=0$. Otherwise, there is only one strategy: move in a martingale
fashion until you reach the endpoints. Thus 
\[
u_{\eps}\left(x\right)=\var\left(X\left(\tau\right)\right)=\e X\left(\tau\right)^{2}-\e X\left(0\right)^{2}=1-x^{2}.
\]
We then have 
\[
\abs{u_{\eps}\left(x\right)-u_{\eps}\left(y\right)}=\abs{1-x^{2}-\left(1-y^{2}\right)}=\abs{x^{2}-y^{2}}=\abs{x-y}\abs{x+y}\leq2\abs{x-y}.
\]
Now assume that the claim is true for any two points on the same $n-1$-dimensional
facet of the hypercube, and let $x,y\in\left[-1,1\right]^{n}$. We
prove the claim for $n$ dimensions by considering the different relations
that $x$ and $y$ can have.
\begin{enumerate}
\item \label{enu:both_points_on_boundary}Suppose both $x,y$ are on the
boundary $\partial\left[-1,1\right]^{n}$. If they are on the same
facet, the claim follows by the induction hypothesis. If they are
on two adjacent facets, we can assume without loss of generality that
$x=\left(1,x_{2},x_{3}\ldots,x_{n}\right)$ and $y=\left(y_{1},1,y_{3},\ldots,y_{n}\right)$.
Then the point $z=\left(1,1,x_{3},\ldots,x_{n}\right)$ shares a facet
with $x$ and with $y$ and satisfies $\norm{x-y}_{1}=\norm{x-z}_{1}+\norm{y-z}_{1}$.
We then have
\begin{align*}
\abs{u_{\eps}\left(x\right)-u_{\eps}\left(y\right)} & \leq\abs{u_{\eps}\left(x\right)-u_{\eps}\left(z\right)}+\abs{u_{\eps}\left(z\right)-u_{\eps}\left(y\right)}\\
 & \leq L_{n-1}\norm{x-z}_{1}+L_{n-1}\norm{z-y}_{1}\\
 & =L_{n-1}\norm{x-y}_{1}.
\end{align*}
If $x$ and $y$ are on opposite facets, then $\norm{x-y}_{1}\geq2$,
and since $0\leq u_{\eps}\leq n$, we have 
\[
\abs{u_{\eps}\left(x\right)-u_{\eps}\left(y\right)}\leq n\leq\frac{n}{2}\norm{x-y}_{1}\leq L_{n}\norm{x-y}_{1}.
\]
\item \label{enu:one_point_inside_one_in_boundary}Suppose that $x\in\left(-1,1\right)^{n}$
and $y$ is its projection onto the boundary of the hypercube in some
direction $j\in\left[n\right]$. Without loss of generality, we can
assume $j=1$ and  $y=\left(1,x_{2},\ldots,x_{n}\right)$. If $u_{\eps}\left(y\right)=0$,
let $T$ be the strategy that always picks direction $1$, and let
$X\left(t\right)$ be the axis-aligned jump process which starts at
$x$ and is moved by $T$. Then 
\begin{align*}
C\left(x,T\right) & =\frac{1+x}{2}u_{\eps}\left(y\right)+\frac{1-x}{2}u_{\eps}\left(-1,x_{2},\ldots x_{n}\right)+1-x_{1}^{2}\\
 & \overset{u_{\eps}\leq n-1\text{ on the boundary}}{\leq}\frac{1-x}{2}\left(n-1\right)+1-x_{1}^{2}\\
 & \overset{\left(\abs x\leq1\right)}{\leq}\frac{n+3}{2}\left(1-x\right)=L_{n}\norm{x-y}_{1}.
\end{align*}
Thus 
\[
\abs{u_{\eps}\left(x\right)-u_{\eps}\left(y\right)}=u_{\eps}\left(x\right)\leq C\left(x,T\right)\leq L_{n}\norm{x-y}_{1}.
\]
If $u_{\eps}\left(y\right)\neq0$, then let $S$ be an optimal strategy,
and let $Y\left(t\right)$ be the corresponding axis-aligned jump
process which starts at $y$. Let $T$ be a strategy which repeats
the choices that $S$ makes on $y$'s facet if $x$ in the interior:
\[
T\left(x_{1},x_{2},\ldots,x_{n}\right)=\begin{cases}
S\left(1,x_{2},\ldots,x_{n}\right) & x\in\left(-1,1\right)^{n}\\
S\left(x\right) & x\in\partial\left[-1,1\right]^{n}.
\end{cases}
\]
Let $X\left(t\right)$ be the corresponding axis-aligned jump process
which starts at $x$. Couple $X\left(t\right)$ and $Y\left(t\right)$
so that they move together, and let $\tau$ be the first time that
$Y\left(t\right)$ hits a facet $j\in2,\ldots,n$. Since $u_{\eps}\left(y\right)\neq0$,
we have $\tau<\infty$. Necessarily, $X\left(t\right)$ will also
hit the same facet at this time (see Figure \ref{fig:x_follows_y}),
and so using the boundary condition formulation, we have $C\left(x,T\right)=\sum_{i=1}^{n}\var\left(X_{i}\left(\tau\right)\right)+\e u_{\eps}\left(X\left(\tau\right)\right)$
and $C\left(y,S\right)=\sum_{i=1}^{n}\var\left(Y_{i}\left(\tau\right)\right)+\e u_{\eps}\left(Y\left(\tau\right)\right)$.
The processes $X\left(t\right)$ and $Y\left(t\right)$ are coupled
and only move in directions $2,\ldots,n$, and so we have 
\[
C\left(x,T\right)-C\left(y,S\right)=\e u_{\eps}\left(X\left(\tau\right)\right)-\e u_{\eps}\left(Y\left(\tau\right)\right).
\]
Since $S$ is optimal but $T$ might be non-optimal, we get 
\begin{align*}
u_{\eps}\left(x\right)-u_{\eps}\left(y\right) & \leq C\left(x,T\right)-C\left(y,S\right)\\
 & =\e u_{\eps}\left(X\left(\tau\right)\right)-\e u_{\eps}\left(Y\left(\tau\right)\right)\\
 & \leq L_{n-1}\e\norm{X\left(\tau\right)-Y\left(\tau\right)}_{1}=L_{n-1}\norm{x-y}_{1}.
\end{align*}
We use a similar technique to get the opposite inequality. Let $S$
be an optimal strategy, and let $X\left(t\right)$ be the corresponding
axis-aligned jump process which starts at $x$. Let $Y\left(t\right)$
start at $y$ and be moved by a strategy $T$ defined as follows.
The process $Y\left(t\right)$ and strategy $T$ are coupled with
$X\left(t\right)$ and $S$ so that whenever $S$ moves $X\left(t\right)$
in directions $2,\ldots,n$, $T$ moves $Y\left(t\right)$ in the
same direction; when $S$ moves $X\left(t\right)$ in direction $1$,
$T$ does nothing. Let $\tau$ be the first time that $X\left(t\right)$
reaches the boundary of the hypercube. After time $\tau$, $T$ moves
$Y\left(t\right)$ optimally. The coordinates $2,\ldots,n$ of $X\left(t\right)$
and $Y\left(t\right)$ are always the same for $t\leq\tau$ (see Figure
\ref{fig:y_follows_x}). Since $S$ is optimal but $T$ might be non-optimal,
we have 
\begin{align*}
u_{\eps}\left(y\right)-u_{\eps}\left(x\right) & \leq C\left(y,T\right)-C\left(x,S\right)\\
 & =\sum_{i=1}^{n}\left(\var\left(Y_{i}\left(\tau\right)\right)-\var\left(X_{i}\left(\tau\right)\right)\right)+\e u_{\eps}\left(Y\left(\tau\right)\right)-\e u_{\eps}\left(X\left(\tau\right)\right)\\
 & =-\var\left(X_{1}\left(\tau\right)\right)+\e u_{\eps}\left(Y\left(\tau\right)\right)-\e u_{\eps}\left(X\left(\tau\right)\right)\\
 & \leq\e u_{\eps}\left(Y\left(\tau\right)\right)-\e u_{\eps}\left(X\left(\tau\right)\right)\overset{\text{Item \eqref{enu:both_points_on_boundary}}}{\leq}L_{n}\e\norm{X\left(\tau\right)-Y\left(\tau\right)}_{1},
\end{align*}
where for the last inequality we use the fact that at time $\tau$
both points are on the boundary of the hypercube. Since $\norm{X\left(\tau\right)-Y\left(\tau\right)}_{1}=1-X_{1}\left(\tau\right)$,
we have $\e\norm{X\left(\tau\right)-Y\left(\tau\right)}_{1}=1-x=\norm{x-y}_{1}$,
and the result follows. 
\begin{figure}
\begin{centering}
\subfloat[\label{fig:x_follows_y}]{\includegraphics[scale=0.5]{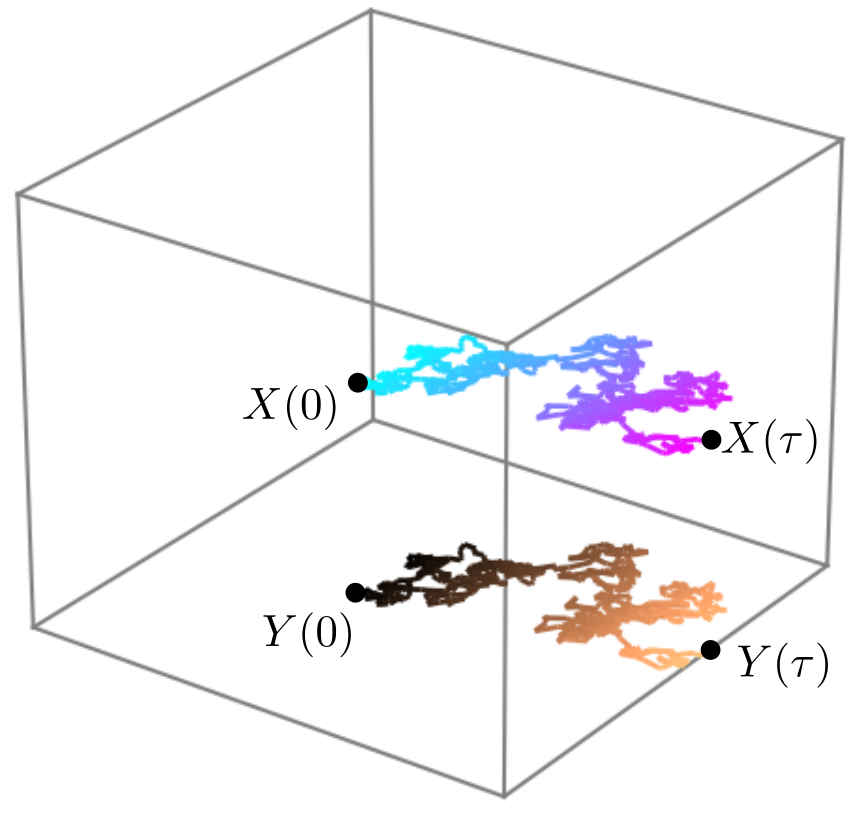}

}\subfloat[\label{fig:y_follows_x}]{\includegraphics[scale=0.5]{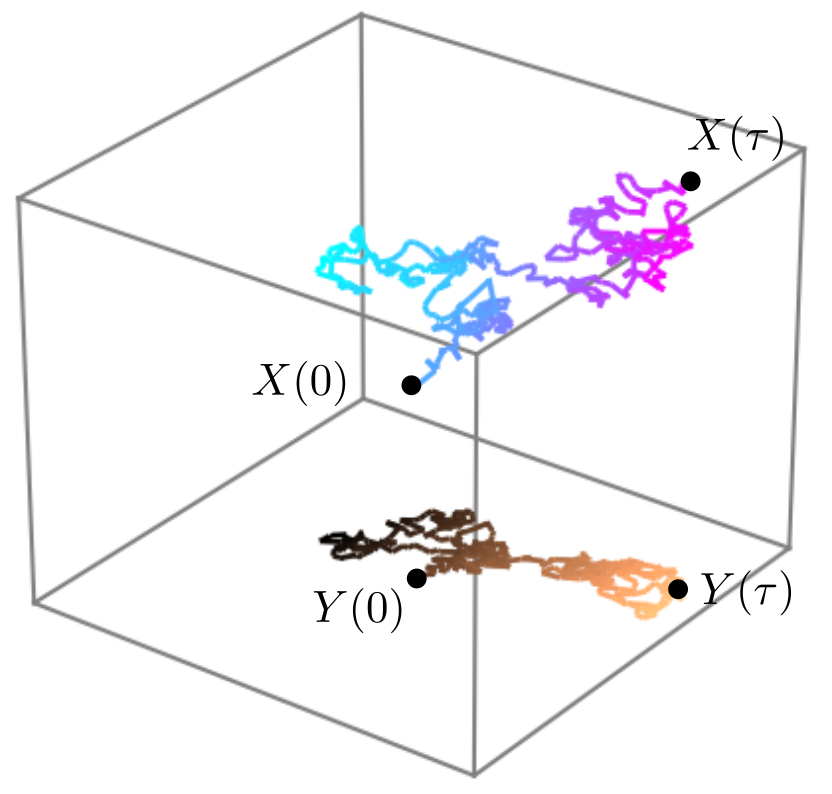}

}
\par\end{centering}
\caption{\ref{fig:x_follows_y}. $X\left(t\right)$ follows the strategy of
$Y\left(t\right)$ until $Y\left(t\right)$ hits a facet. At time
$\tau$, both are on the same facet. \ref{fig:y_follows_x}. $Y\left(t\right)$
follows the strategy of $X\left(t\right)$. At time $\tau$, both
processes are on the boundary. }

\end{figure}
\item Suppose that $x,y\in\left(-1,1\right)^{n}$ differ by only one coordinate.
Let $S$ be an optimal Markov strategy which moves $X\left(t\right)$.
Let $T$ be the following strategy, which is coupled with $X\left(t\right)$:
\[
T=\begin{cases}
S\left(X\left(t\right)\right) & \text{if both \ensuremath{X\left(t\right)} and \ensuremath{Y\left(t\right)} are in the interior of the hypercube}\\
S\left(Y\left(t\right)\right) & \text{otherwise.}
\end{cases}
\]
Let $\tau$ be the first time that either $X\left(t\right)$ or $Y\left(t\right)$
hit the boundary of the hypercube. Up until time $\tau$ they move
together, and so the difference in running cost between them is $0$,
and the distance between them stays constant. Since $S$ is optimal
but $T$ might be non-optimal, and since for $t>\tau$ both processes
move according to an optimal strategy, we have 
\begin{align*}
u_{\eps}\left(y\right)-u_{\eps}\left(x\right) & \leq C\left(y,Y\right)-C\left(x,S\right)\\
 & =\e u_{\eps}\left(Y\left(\tau\right)\right)-\e u_{\eps}\left(X\left(\tau\right)\right)\\
 & \leq L_{n}\e\norm{X\left(\tau\right)-Y\left(\tau\right)}_{1}=L_{n}\norm{x-y}_{1},
\end{align*}
where for the last inequality, we use the fact that at time $\tau$,
either both points are on the same facet, in which case the the desired
inequality follows from the induction hypothesis, or one point is
the projection of the other, in which case the desired inequality
follows from Item (\ref{enu:one_point_inside_one_in_boundary}) above.
Switching the roles of $x$ and $y$ gives the opposite inequality.
\item Finally, for any two general points $x,y\in\left[-1,1\right]^{n}$,
consider the points $z_{i}=\left(x_{1},\ldots,x_{n-i},y_{n-i+1},\ldots,y_{n}\right)$
for $i=1,\ldots,n$. Set also $z_{0}=x$. Then
\begin{align*}
\abs{u_{\eps}\left(x\right)-u_{\eps}\left(y\right)} & \leq\sum_{i=1}^{n}\abs{u_{\eps}\left(z_{i}\right)-u_{\eps}\left(z_{i-1}\right)}\\
 & \leq\sum_{i=1}^{n}L_{n}\abs{z_{i}-z_{i-1}}\\
 & \leq L_{n}\norm{x-y}_{1}.
\end{align*}
 
\end{enumerate}
\end{proof}
\begin{rem}
The Lipschitz constant in Lemma \ref{lem:u_eps_are_lipschitz} is
tight up to an additive constant: consider the function 
\[
f\left(x\right)=\begin{cases}
1 & x_{1}=1\\
\prod_{i=2}^{n}x_{i} & x_{1}=-1.
\end{cases}
\]
Then $u_{\eps}\left(1,0,\ldots,0\right)=0$ while $u_{\eps}\left(-1,0,\ldots,0\right)=n-1$. 
\end{rem}

It is natural to look at the behavior of $u_{\eps}$ as $\eps\to0$.
Is there a ``limiting strategy'' in some sense, that works for smaller
and smaller step-sizes? At the very least, there is a limiting cost:
\begin{prop}
\label{prop:limiting_cost_of_axis_aligned}There exists a function
$u:\left[-1,1\right]^{n}\to\r$ such that for every $x\in\left[-1,1\right]^{n}$,
$u_{\eps}\left(x\right)\to u\left(x\right)$ as $\eps\to0$ dyadically
(i.e. we consider numbers of the form $\eps=2^{-k}$ as $k\to\infty$).
\end{prop}

\begin{proof}
By the recursive construction of $u_{\eps}$, it suffices to prove
this for $x\in\left(-1,1\right)^{n}$. For every such $x$, we have
$u_{\eps/2}\left(x\right)\leq u_{\eps}\left(x\right)$ : an $\eps$-strategy
$S$ can be coupled with an $\eps/2$-strategy in the following manner.
If $S$ picks coordinate $i$ at time $t$, so that $X_{i_{t}}\left(t+1\right)\in\left\{ a,b\right\} $,
the simulating strategy can repeatedly pick coordinate $i$ until
$X_{i}$ takes one of the values $\left\{ a,b\right\} $; the randomness
can be coupled so that the same value is reached. Since the processes
are martingales, the escape probabilities and the expected sum square
of jumps are equal. So $u_{\eps}\left(x\right)$ is non-negative and
monotone decreasing as $\eps=2^{-k}\to0$, and must converge. 
\end{proof}
In fact, the following lemma is a consequence of the Arzela-Ascoli
theorem and Lemma \ref{lem:u_eps_are_lipschitz}:
\begin{lem}
\label{lem:uniform_convergence}$\norm{u_{\eps}-u}_{\infty}\to0$
as $\eps\to0$ dyadically. The function $u$ is $L_{n}$-Lipschitz. 
\end{lem}

The function $u$ can always be bounded from above by considering
strategies for $u_{\eps}$:
\begin{fact}
\label{fact:dynamic_programming_inequality_for_u}For every $x\in\left(-1,1\right)^{n}$,
every dyadic $\alpha>0$ small enough, and every direction $i$, 
\[
u\left(x\right)\leq\frac{u\left(x+\alpha e_{i}\right)+u\left(x-\alpha e_{i}\right)}{2}+\alpha^{2}.
\]
\end{fact}

\begin{proof}
Let $0<\eps\leq\alpha$ be dyadic, and consider the following strategy
for $u_{\eps}$ when starting at $x$: always choose to go in direction
$i$ until you reach either $x-\alpha e_{i}$ or $x+\alpha e_{i}$
(we assume $\alpha$ is small enough so that both points are in the
interior $\left(-1,1\right)^{n}$); afterwards, continue optimally.
There is equal probability of hitting either $x-\alpha e_{i}$ or
$x+\alpha e_{i}$, and the running cost for doing so is $\alpha^{2}$.
Thus
\begin{align*}
u\left(x\right) & \leq u_{\eps}\left(x\right)\\
 & \leq\frac{u_{\eps}\left(x+\alpha e_{i}\right)+u_{\eps}\left(x+\alpha e_{i}\right)}{2}+\alpha^{2}.
\end{align*}
Taking the limit $\eps\to0$ gives the result.
\end{proof}
\begin{rem}
\label{rem:intuition_behind_axis_aligned_laplacian}The intuition
behind the operator $\axis=\min_{k}\frac{\partial^{2}}{\partial x_{k}^{2}}$
is as follows. Consider the dynamic programming equation (\ref{eq:dynamic_programming_equation}):
\[
u_{\eps}\left(x\right)=\min_{i}\frac{u_{\eps}\left(x+\eps e_{i}\right)+u_{\eps}\left(x+\eps e_{i}\right)}{2}+\eps^{2}.
\]
Since there are only finitely many indices, there is an index $k$
which appears infinitely many times as $\eps\to0$. For this particular
$k$, taking the limit $\eps\to0$, we have 
\[
-2=\lim_{\eps\to0}\frac{u_{\eps}\left(x+\eps e_{k}\right)-2u_{\eps}\left(x\right)+u_{\eps}\left(x+\eps e_{k}\right)}{\eps^{2}}.
\]
If we could replace $u_{\eps}$ in the above equation by $u$, and
if we knew that $u$ was twice-differentiable, the right hand side
would equal the second derivative of $u$, and we would get $\frac{\partial^{2}u}{\partial x_{k}^{2}}=-2$.
Since the index $k$ was chosen as the minimum, we would hope to reach
the following partial differential equation for $u$: 
\[
\axis u:=\min_{k}\frac{\partial^{2}u}{\partial x_{k}^{2}}=-2.
\]
\end{rem}

\begin{thm}
\label{thm:obeying_axis_aligned_laplacian}Either $u\equiv0$, or
it satisfies $\axis u=-2$ for all $x\in\left(-1,1\right)^{n}$ in
the viscosity sense. 
\end{thm}

The proof uses standard techniques (see e.g \cite[Chapter 3]{blanc_rossi_game_theory_and_pde}),
and relies on the uniform convergence property proved above. 
\begin{proof}
We start by showing that $u$ is a viscosity subsolution (recall Definition
\ref{def:viscosity_solutions}). Let $x_{0}\in\left(-1,1\right)^{n}$,
and let $\varphi:\left[-1,1\right]^{n}\to\r$ be a smooth function
so that $\varphi-u$ has a minimum at $x_{0}$, with $\varphi\left(x_{0}\right)=u\left(x_{0}\right)$.
Let $k$ be the direction for which $\frac{\partial^{2}}{\partial x_{k}^{2}}\varphi\left(x_{0}\right)$
is minimal. By Fact \ref{fact:dynamic_programming_inequality_for_u},
for every dyadic $\alpha$ small enough we have
\begin{align*}
\varphi\left(x_{0}\right) & =u\left(x_{0}\right)\\
 & \overset{\text{Fact \ref{fact:dynamic_programming_inequality_for_u}}}{\leq}\frac{u\left(x_{0}+\alpha e_{k}\right)+u\left(x_{0}-\alpha e_{k}\right)}{2}+\alpha^{2}\\
 & \leq\frac{\varphi\left(x_{0}+\alpha e_{k}\right)+\varphi\left(x_{0}-\alpha e_{k}\right)}{2}+\alpha^{2}\\
 & =\varphi\left(x_{0}\right)+\alpha^{2}\frac{1}{2}\frac{\partial^{2}}{\partial x_{k}^{2}}\varphi\left(x_{0}\right)+o\left(\alpha^{2}\right)+\alpha^{2}.
\end{align*}
Rearranging, dividing by $\alpha^{2}$, and taking the limit $\alpha\to0$
gives
\[
\frac{\partial^{2}}{\partial x_{k}^{2}}\varphi\left(x_{0}\right)\geq-2.
\]
Since $k$ was chosen to be the direction for which $\frac{\partial^{2}}{\partial x_{k}^{2}}\varphi\left(x_{0}\right)$
is minimal, we get that 
\[
\axis\varphi\left(x_{0}\right)\geq-2,
\]
which means that $u$ is a subsolution to the Dirichlet problem. Note
that we did not uniform convergence of $u_{\eps}$ to $u$ to prove
this; we only used the monotone pointwise convergence of $u_{\eps}$
to $u$, which is used in Fact \ref{fact:dynamic_programming_inequality_for_u}.
The dynamic programming equation for $u_{\eps}$ makes it relatively
easy to bound from above. 

For showing that $u$ is a supersolution, let $\varphi:\left[-1,1\right]^{n}\to\r$
be a smooth function so that $u-\varphi$ has a strict minimum at
$x_{0}$. Denote $\psi_{\eps}=u_{\eps}-\varphi$, and $\psi=\lim_{\eps\to0}\psi_{\eps}=u-\varphi$.
By Lemma \ref{lem:uniform_convergence}, we also have uniform convergence
of $\psi_{\eps}$ to $\psi$:
\[
\norm{\psi_{\eps}-\psi}_{\infty}=\norm{u_{\eps}-\varphi-\left(u-\varphi\right)}_{\infty}=\norm{u_{\eps}-u}_{\infty}\to0.
\]
Let $x_{\eps}=\mathrm{argmin}_{x\in\left[-1,1\right]^{n}}\left\{ \psi_{\eps}\right\} $;
the minimum exists since $\psi_{\eps}$ is continuous for every $\eps$.
By definition, for all $x\in\left[-1,1\right]^{n}$, 
\begin{equation}
u_{\eps}\left(x\right)-\varphi\left(x\right)\geq u_{\eps}\left(x_{\eps}\right)-\varphi\left(x_{\eps}\right).\label{eq:the_mythical_x_eps}
\end{equation}
We'll now show that $x_{\eps}\to x_{0}$, where $x_{0}$ is the minimizer
of $\psi\left(x\right)=u\left(x\right)-\varphi\left(x\right)$. Suppose
not. Since $\left[-1,1\right]^{n}$ is compact, there exists a subsequence,
which we still call $x_{\eps}$, which converges to some $z\neq x_{0}$.
Now, since $x_{0}$ is a strict minimizer of $\psi$, denote 
\begin{equation}
0<\Delta:=\psi\left(z\right)-\psi\left(x_{0}\right).\label{eq:definition_of_psi_delta}
\end{equation}
Since $\psi_{\eps}\to\psi$ uniformly, for small enough $\eps$ we
have: 
\begin{equation}
\psi_{\eps}\left(x_{0}\right)<\psi\left(x_{0}\right)+\frac{\Delta}{3}\label{eq:psi_is_smaller_than_something}
\end{equation}
and 
\begin{equation}
\psi_{\eps}\left(x_{\eps}\right)>\psi\left(x_{\eps}\right)-\frac{\Delta}{3}.\label{eq:psi_is_larger_than_something}
\end{equation}
Since $\psi$ is continuous and $x_{\eps}\to z$, for small enough
$\eps$ we have 
\begin{equation}
\psi\left(x_{\eps}\right)>\psi\left(z\right)-\frac{\Delta}{3}.\label{eq:psi_and_z}
\end{equation}
Now, on one hand,
\begin{equation}
\psi_{\eps}\left(x_{\eps}\right)\overset{\eqref{eq:psi_is_larger_than_something}}{>}\psi\left(x_{\eps}\right)-\frac{\Delta}{3}\overset{\eqref{eq:psi_and_z}}{>}\psi\left(z\right)-\frac{2\Delta}{3}.\label{eq:psi_almost_there}
\end{equation}
On the other hand, 
\[
\psi_{\eps}\left(x_{0}\right)\overset{\eqref{eq:psi_is_smaller_than_something}}{<}\psi\left(x_{0}\right)+\frac{\Delta}{3}\overset{\eqref{eq:definition_of_psi_delta}}{=}\psi\left(z\right)-\Delta+\frac{\Delta}{3}=\psi\left(z\right)-\frac{2\Delta}{3}\overset{\eqref{eq:psi_almost_there}}{<}\psi_{\eps}\left(x_{\eps}\right).
\]
Thus $\psi_{\eps}\left(x_{0}\right)<\psi_{\eps}\left(x_{\eps}\right)$,
a contradiction to the minimality of $x_{\eps}$.

For small enough $\eps$, since $x_{0}\in\left(-1,1\right)^{n}$ and
$x_{\eps}\to x_{0}$, we have that $x_{\eps}\pm\eps e_{i}\in\left(-1,1\right)^{n}$
for all directions $i\in\left[n\right]$. By the dynamic programming
equation (\ref{eq:dynamic_programming_equation}), for every such
$\eps$, the value of function $u_{\eps}$ evaluated at $x_{\eps}$
is given by 
\[
-\eps^{2}=\min_{i}\frac{u_{\eps}\left(x_{\eps}+\eps e_{i}\right)+u_{\eps}\left(x_{\eps}-\eps e_{i}\right)}{2}-u_{\eps}\left(x_{\eps}\right).
\]
Since there are only finitely many directions $i\in\left[n\right]$,
there is some direction $k\in\left[n\right]$ which appears infinitely
many times as the minimizer in the above equation. Restricting ourselves
just to those $\eps$'s, we have 
\begin{align}
-\eps^{2} & =\frac{u_{\eps}\left(x_{\eps}+\eps e_{k}\right)+u_{\eps}\left(x_{\eps}-\eps e_{k}\right)}{2}-u_{\eps}\left(x_{\eps}\right)\nonumber \\
 & =\frac{1}{2}\left(u_{\eps}\left(x_{\eps}+\eps e_{k}\right)-u_{\eps}\left(x_{\eps}\right)\right)+\frac{1}{2}\left(u_{\eps}\left(x_{\eps}-\eps e_{k}\right)-u_{\eps}\left(x_{\eps}\right)\right).\label{eq:dynamic_programming_again}
\end{align}
By applying (\ref{eq:the_mythical_x_eps}), the first expression in
the parenthesis on the right hand side can be bounded by:
\begin{align*}
u_{\eps}\left(x_{\eps}+\eps e_{k}\right)-u_{\eps}\left(x_{\eps}\right) & =u_{\eps}\left(x_{\eps}+\eps e_{k}\right)-\varphi\left(x_{\eps}+\eps e_{k}\right)-u_{\eps}\left(x_{\eps}\right)+\varphi\left(x_{\eps}+\eps e_{k}\right)\\
 & \geq u_{\eps}\left(x_{\eps}\right)-\varphi\left(x_{\eps}\right)-u_{\eps}\left(x_{\eps}\right)+\varphi\left(x_{\eps}+\eps e_{k}\right)\\
 & =\varphi\left(x_{\eps}+\eps e_{k}\right)-\varphi\left(x_{\eps}\right).
\end{align*}
Similarly, 
\[
u_{\eps}\left(x_{\eps}-\eps e_{k}\right)-u_{\eps}\left(x_{\eps}\right)\geq\varphi\left(x_{\eps}-\eps e_{k}\right)-\varphi\left(x_{\eps}\right).
\]
Plugging this back into (\ref{eq:dynamic_programming_again}), we
have
\begin{align*}
-\eps^{2} & \geq\frac{\varphi\left(x_{\eps}+\eps e_{k}\right)+\varphi\left(x_{\eps}-\eps e_{k}\right)}{2}-\varphi\left(x_{\eps}\right)\\
 & =\frac{1}{2}\eps^{2}\frac{\partial^{2}}{\partial x_{k}^{2}}\varphi\left(x_{\eps}\right)+o\left(\eps^{2}\right).
\end{align*}
Dividing by $\eps^{2}$, taking the limit $\eps\to0$ gives us 
\[
\frac{\partial^{2}}{\partial x_{k}^{2}}\varphi\left(x_{0}\right)\leq-2.
\]
Since this is true for some $k$, it is true in particular for the
smallest second derivative of $\varphi$. Thus 
\[
\axis\varphi\left(x_{0}\right)\leq-2,
\]
which means that $u$ is a supersolution to the Dirichlet boundary
problem.
\end{proof}
The axis-aligned Laplacian is a non-linear operator. It is, however,
monotone in the Hessian $\grad^{2}u$.
\begin{fact}
\label{fact:axis_aligned_laplacian_is_monotone}Let $u,v$ be twice-differentiable
functions. If $\grad^{2}u\geq\grad^{2}v$ (i.e. the matrix $\grad^{2}u-\grad^{2}v$
is positive semidefinite), then $\axis u\geq\axis v$.
\end{fact}

\begin{proof}
If $\grad^{2}\left(u-v\right)$ is positive semidefinite, then for
all $i\in\left[n\right]$, 
\[
0\leq\left\langle e_{i},\grad^{2}\left(u-v\right)e_{i}\right\rangle =\frac{\partial^{2}}{\partial x_{i}^{2}}u-\frac{\partial^{2}}{\partial x_{i}^{2}}v.
\]
In particular, for $i^{*}=\mathrm{argmin}_{i}\frac{\partial^{2}}{\partial x_{i}^{2}}u$
we get $\axis u\geq\frac{\partial^{2}}{\partial x_{i^{*}}^{2}}v\geq\axis v$.
\end{proof}
This gives hope that solutions to the Dirichlet boundary-value problem
$\axis u\left(x\right)=f\left(x\right)$ are unique. This is indeed
true, if the function $f$ does not change sign, and follows from
a general comparison principle. 
\begin{thm}
\label{thm:comparison_principle}Let $\Omega\in\r^{n}$ be a bounded
domain. Let $f:\Omega\to\r$ be continuous with $\sup_{\Omega}f<0$.
Suppose that $u_{1},u_{2}$ are continuous functions such that $\axis u_{1}\leq f\leq\axis u_{2}$
in $\Omega$ and $u_{1}\geq u_{2}$ on $\partial\Omega$. Then $u_{1}\geq u_{2}$
in $\Omega$. 
\end{thm}

Theorems \ref{thm:obeying_axis_aligned_laplacian} and \ref{thm:comparison_principle}
immediately yield Theorem \ref{thm:intro_axis_aligned_laplacian}
as a corollary. 

The proof of Theorem \ref{thm:comparison_principle} follows the scheme
of Lu and Wang \cite{lu_wang_inhomogeneous_laplace_equation}. We
start by showing a strict comparison principle.
\begin{lem}
\label{lem:strict_comparison_principle}Let $f_{1},f_{2}:\Omega\to\r$
be continuous functions such that $f_{1}<f_{2}$. Let $u_{1},u_{2}$
be continuous functions such that $\axis u_{1}\leq f_{1}<f_{2}\leq\axis u_{2}$
in $\Omega$ and $u_{1}\geq u_{2}$ on $\partial\Omega$. Then $u_{1}\geq u_{2}$
in $\Omega$. 
\end{lem}

\begin{proof}
Suppose for the sake of contradiction that there exists $x^{*}\in\Omega$
such that $u_{1}\left(x^{*}\right)<u_{2}\left(x^{*}\right)$. Let
$\eps>0$, define $M_{\eps}=\max_{x,y\in\Omega^{2}}u_{2}\left(x\right)-u_{1}\left(y\right)-\frac{1}{2\eps}\norm{x-y}_{2}^{2}$,
and let $x_{\eps}$ and $y_{\eps}$ be the maximizers. By Lemma 3.1
of \cite{crandall_ishii_lions_users_guide_to_viscosity}, $\lim_{\eps\to0}M_{\eps}=\max_{\bar{\Omega}}\left\{ u_{2}-u_{1}\right\} >0$,
and $\lim_{\eps\to0}\norm{x_{\eps}-y_{\eps}}_{2}=0$. By Lemma 3.2
of \cite{crandall_ishii_lions_users_guide_to_viscosity}, there exist
symmetric $n\times n$ matrices $X$ and $Y$ with $X\leq Y$ such
that:
\begin{enumerate}
\item There exists a sequence of smooth functions $\varphi_{k}$ and points
$x_{k}$ such that $\varphi_{k}-u_{2}$ has a local minimum at $x_{k}$,
$x_{k}\to x_{\eps}$, and $\grad^{2}\varphi_{k}\left(x_{k}\right)\to X$.
\item There exists a sequence of smooth functions $\psi_{k}$ and points
$y_{k}$ such that $\psi_{k}-u_{1}$ has a local maximum at $y_{k}$,
$y_{k}\to y_{\eps}$, and $\grad^{2}\psi_{k}\left(y_{k}\right)\to Y$. 
\end{enumerate}
By Definition \ref{def:viscosity_solutions} of viscosity solutions,
$f_{2}\left(x_{k}\right)\leq\axis\varphi_{k}\left(x_{k}\right)$ and
$f_{1}\left(y_{k}\right)\geq\axis\psi_{k}\left(y_{k}\right)$ for
every $k$. Taking the limit $k\to\infty$, since $\min_{i}X_{ii}$,
$f_{1}$ and $f_{2}$ are all continuous functions, we have 
\begin{equation}
f_{2}\left(x_{\eps}\right)\leq\min_{i}X_{ii}\overset{\text{Fact \ref{fact:axis_aligned_laplacian_is_monotone}}}{\leq}\min_{i}Y_{ii}\leq f_{1}\left(y_{\eps}\right).\label{eq:strict_comparison_different_points}
\end{equation}
Since $\Omega$ is bounded and the maximum of $u_{2}-u_{1}$ is obtained
in the interior of $\Omega$, there are subsequences of $x_{\eps}$
and $y_{\eps}$, which we also denote $x_{\eps}$ and $y_{\eps}$,
which converge to some $x_{0}\in\Omega$. Sending $\eps\to0$, continuity
of $f_{1},f_{2}$ together with equation (\ref{eq:strict_comparison_different_points})
then give
\[
f_{2}\left(x_{0}\right)\leq f_{1}\left(x_{0}\right),
\]
contradicting the assumption of the lemma. 
\end{proof}
\begin{proof}[Proof of Theorem \ref{thm:comparison_principle}]
For any $0<\delta<1$, define $u_{\delta}=\left(1-\delta\right)u_{2}-\delta\max_{\partial\Omega}\abs{u_{2}}$.
Since $\axis\alpha f=\alpha\axis f$ for all $\alpha>0$, we have
$\axis u_{\delta}=\left(1-\delta\right)\axis u_{2}\geq\left(1-\delta\right)f>f\geq\axis u_{1}$.
Also, $u_{\delta}\leq u_{2}\leq u_{1}$ on $\partial\Omega$. By Lemma
\ref{lem:strict_comparison_principle}, $u_{1}\geq u_{\delta}$ in
$\Omega$. Sending $\delta\to0$ gives the desired result. 
\end{proof}

\section{An example: The OR function\label{sec:or_example}}

\subsection{$n$ bits}

Let $f\left(x\right)$ be the $n$-bit OR function, which returns
$1$ if and only if one of the bits $x_{i}$ is equal to $1$. This
is a symmetric function - it depends only on the number of bits in
the input. There is thus essentially only one decision tree algorithm:
read bits at random until the value of the function is computed. Every
bit has value $1$ with probability $1/2$, so apart from the $n$-th
read bit, every bit has a probability of $1/2$ of ending the computation.
Thus, as $n\to\infty$, the number of bits queried tends towards a
geometric random variable with parameter $2$, and the expected number
of bits queried is $2$. 

For fractional query algorithms, there are many more algorithms to
choose from. Since the OR function needs only a single bit to be set
to $1$, the natural algorithm is to always update the largest bit.
Updating smaller bits is intuitively wasteful, because either they
reach the value $-1$ (in which case the other bits need to be evaluated
anyway), or they take longer to reach the value $1$ than the largest
bit. This intuition holds true in the axis-aligned algorithm setting.
\begin{thm}
Let $\eps=2^{-k}$, and let $S_{\mathrm{max}}:\left[-1,1\right]^{n}\to\left[n\right]$
be given by $S_{\mathrm{max}}\left(x\right)=\mathrm{argmax}_{i}\left(x\right)$.
Then $S_{\mathrm{max}}$ is optimal, i.e. for all other decision strategies
$T$, we have $C\left(S_{\mathrm{max}}\right)\leq C\left(T\right)$. 
\end{thm}

\begin{proof}[Proof sketch]
Given a strategy $T$ with process $X\left(t\right)$, define a new
strategy $Q=Q\left(T\right)$ with process $Y\left(t\right)$, which
runs according to $T$ up to the first time $t_{1}$ that $T$ doesn't
update the largest bit. Denote a largest bit by $i_{0}$ and let $\alpha=X_{i_{0}}\left(t_{1}\right)$.
$Q$ does update this bit, and then makes the same choices that $T$
would make assuming the largest bit wasn't updated. It continues so
until the time $t_{2}$ where $T$ finally updates the original bit,
or another bit of value $\alpha$. From this point on, it again runs
according to $T$. 

The claim is that $Q$ always does better than $S$ in terms of expected
runtime: until time $t_{1}$, both strategies give identical processes.
During the interval $t_{1}<t<t_{2}$, $T$ cannot finish before $Q$
does: in order to get $f\left(x\right)=-1$, all bits need to be updated,
and in particular bit $i_{0}$; in order to get $f\left(x\right)=1$,
a bit with value $\alpha$ needs to be updated. After time $t_{2}$,
if $Q$ didn't already finish, we have that $X\left(t\right)$ and
$Y\left(t\right)$ are permutations of each other, and so in expectation
$Q$ and $T$ perform the same. Thus $C\left(Q\left(T\right)\right)\leq C\left(T\right)$.

Let $S^{*}$ be an optimal strategy. The strategy $S_{m}=Q^{m}\left(S^{*}\right)$
is identical with $S_{\mathrm{max}}$ for the first $m$ steps. We
thus have 
\begin{align*}
C\left(S_{\mathrm{max}}\right) & \leq C\left(S_{m}\right)+\p\left[\text{algorithm runs longer than \ensuremath{m} steps }\right]\cdot n\\
 & =C\left(S^{*}\right)+\p\left[\text{algorithm runs longer than \ensuremath{m} steps }\right]\cdot n.
\end{align*}
The probability that the algorithm runs longer than $m$ steps goes
to $0$ as $m\to\infty$, and we have 
\[
C\left(S_{\mathrm{max}}\right)\leq C\left(S^{*}\right),
\]
so picking the maximum is also an optimal strategy. 
\end{proof}
\begin{rem}
For continuous-time processes, Jacka, Warren and Windridge \cite[Section 6]{jacka_warren_windridge_minimizing_time_to_a_decision}
prove that when $n=2$, picking the maximum entry is the optimal strategy
in a stronger sense: this strategy actually stochastically dominates
all other strategies. They conjecture that the same is true for general
$n$; this conjecture is strengthened by the above result (continuous-time
processes harbor subtle difficulties of measurability when there are
two variables with the same values). 
\end{rem}

\begin{prop}
The best fractional query algorithm needs to query only one bit: $C\left(S^{*}\right)\to1$
as $n\to\infty$.
\end{prop}

\begin{proof}
We give a strategy $S$ whose cost tends to $1$ as $n\to\infty$.
Let $\eps=1/n$. Suppose that at time $t$, all coordinates are equal
to $-t/n$. For every coordinate $i$, run $X_{i}$ until it exits
the interval $\left[-\frac{t+1}{n},1\right]$. We call this a single
iteration. At the end of the iteration, either we have found a bit
whose value is $1$, or all bits are set to $-\frac{t+1}{n}$. The
probability for a single bit to exit at $1$ is $\frac{1}{n+t+1}$,
and the probability to exit at $-\frac{t+1}{n}$ is $\frac{n+t}{n+t+1}$.
The probability that all coordinates exit at $-\left(t+1\right)/n$
is then 
\[
\left(1-\frac{1}{n+t+1}\right)^{n}\leq\left(1-\frac{1}{3n}\right)^{n}\leq e^{-1/3}.
\]
After $t$ iterations, either all bits have value $-t/n$, or one
bit has value $1$ and all other bits have value $-t/n$ or $-\left(t-1\right)/n$.
If we stop after $t$ iterations, the expected cost is then bounded
by 
\[
c_{t}\leq1+\sum_{i=1}^{n}\left(\frac{t}{n}\right)^{2}=1+\frac{t^{2}}{n}.
\]
Denoting $p=e^{-1/3}$, the expected cost of the algorithm then satisfies
\[
C\left(S\right)\leq\sum_{t=1}^{n}\left(1-p\right)p^{t}\left(1+\frac{t^{2}}{n}\right)\leq1+\frac{1}{n}\sum_{t=1}^{n}\left(1-p\right)p^{t}t^{2}.
\]
The last term on the right hand side goes to $0$ as $n\to\infty$,
and we get the desired result. 
\end{proof}
Of course, at least one bit needs to be queried.

\subsection{$2$ bits}

Below we give a derivation of $u\left(x\right)$ for the $2$-bit
parity function. In fact, it suffices to calculate the value of $u\left(x\right)$
on the diagonal $x_{1}=x_{2}$: for all other points, we know what
the optimal direction is, and the value of $u$ is just a linear combination
of the value on the diagonal and the value on the boundary (which
is $0$). For example, for $x_{2}>x_{1}$, we have
\begin{equation}
u\left(x_{1},x_{2}\right)=\frac{x_{2}-x_{1}}{1-x_{1}}\left(0+\left(1-x_{2}\right)^{2}\right)+\frac{1-x_{2}}{1-x_{1}}\left(g\left(x_{1}\right)+\left(x_{2}-x_{1}\right)^{2}\right),\label{eq:u_for_or_2}
\end{equation}
where $g\left(x\right)=u\left(x,x\right)$. To obtain $g$, denote
$g_{\eps}=u_{\eps}\left(x,x\right)$. For axis-aligned processes with
jump size $\eps$, the optimal strategy always picks the largest entry,
and this gives a recurrence relation for $g_{\eps}$, with the following
strategy. Suppose that $X\left(t\right)=\left(x,x\right)$. Update
coordinate $1$, until either $X\left(t\right)=\left(1,x\right)$,
or $X\left(t\right)=\left(x-\eps,x\right)$. Then, update coordinate
$2$, until either $X\left(t\right)=\left(x-\eps,1\right)$, or $X\left(t\right)=\left(x-\eps,x-\eps\right)$.
A short calculation gives the relation
\[
\frac{g_{\eps}\left(x\right)-g_{\eps}\left(x-\eps\right)}{\eps}+\frac{\eps g_{\eps}\left(x\right)}{\left(1-x\right)^{2}}+\frac{2g_{\eps}\left(x\right)}{1-x}=2\left(1-x+\eps\right)+\eps\frac{1-x+\eps}{1-x}.
\]
This suggests that $g\left(x\right)$ satisfies the differential equation
\begin{equation}
g'\left(x\right)+g\left(x\right)\frac{2}{1-x}=2\left(1-x\right).\label{eq:differential_equation_for_diagonal}
\end{equation}
Formally, we have not shown that $\frac{g_{\eps}\left(x\right)-g_{\eps}\left(x-\eps\right)}{\eps}$
converges to $g'\left(x\right)$ as $\eps\to0$, or indeed that $g$
is even differentiable. However, due to the uniqueness guaranteed
by Theorem \ref{thm:intro_axis_aligned_laplacian}, if we find a continuous
function $v$ which satisfies the boundary conditions for OR and has
$\axis v=-2$ in $\left(-1,1\right)^{n}$, then necessarily $u=v$. 

The solution to the differential equation (\ref{eq:differential_equation_for_diagonal})
(with appropriate boundary conditions) is given by 
\begin{equation}
g\left(x\right)=2\left(1-x\right)^{2}\log2+4x\log\left(1-x\right)-2x^{2}\log\left(1-x\right)-2\log\left(1-x\right).\label{eq:g_under_mild_assumption}
\end{equation}
The function $u\left(x_{1},x_{2}\right)$ obtained by plugging in
(\ref{eq:g_under_mild_assumption}) into (\ref{eq:u_for_or_2}) is
given in Figure \ref{fig:or_analytic_solution}. It can be shown to
satisfy $\axis u\left(x_{1},x_{2}\right)=-2$ as needed.

Some sample paths of $X\left(t\right)$ for $\eps=2^{-7}$ are given
in Figure \ref{fig:or_sample_paths}.

\begin{figure}
\begin{centering}
\includegraphics[scale=0.5]{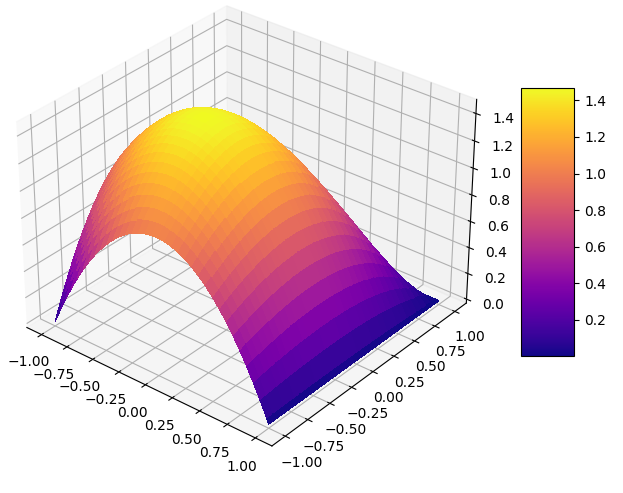}
\par\end{centering}
\caption{\label{fig:or_analytic_solution}The limiting cost $u\left(x,y\right)$
for the OR function on two bits.}
\end{figure}

\subsection{A heuristic for the OR of two functions}

Let $n=n_{1}+n_{2}$ and $g_{i}:\left\{ -1,1\right\} ^{n_{i}}\to\left\{ -1,1\right\} $
be two given functions. Let $f\left(x_{1},x_{2}\right)=OR\left(g_{1}\left(x_{1}\right),g_{2}\left(x_{2}\right)\right)$,
where $x_{i}\in\left\{ -1,1\right\} ^{n_{i}}$. For a given $\eps>0$,
what is the optimal direction-choosing strategy $S$ for $f$?

The following may be a useful heuristic for choosing whether to update
$g_{1}$ or to update $g_{2}$. Let $S_{1}$ and $S_{2}$ be the optimal
strategies for $g_{1}$ and $g_{2}$, and let $u_{1}\left(x\right)$
and $u_{2}\left(y\right)$ be the costs for running $S_{1}$ on $x$
and $S_{2}$ on $y$, respectively. Consider the (definitely non-optimal)
strategy which chooses a function $g_{i}$ and completely evaluates
it using $S_{i}$. The expected cost for doing so is $u_{i}\left(x_{i}\right)$.
If eventually $g_{i}=1$ (which happens with probability $\frac{1+g_{i}\left(x_{i}\right)}{2}$),
then the algorithm is finished. Otherwise, the other function needs
to be computed. In total, the cost for picking $g_{i}$ this way is,
for $\left\{ i,j\right\} =\left\{ 1,2\right\} $, 
\[
c\left(g_{i}\right)=u_{i}\left(x_{i}\right)+\frac{1-g_{i}\left(x_{i}\right)}{2}u_{j}\left(x_{j}\right).
\]
A short calculation shows that this simple strategy should then pick
$i=1$ (i.e. $c\left(g_{1}\right)\leq c\left(g_{2}\right)$) only
if 
\begin{equation}
\frac{u_{1}\left(x_{1}\right)}{1+g_{1}\left(x_{1}\right)}\leq\frac{u_{2}\left(x_{2}\right)}{1+g_{2}\left(x_{2}\right)}.\label{eq:or_heuristic}
\end{equation}
This suggests a heuristic for functions of the form $f=OR\left(g_{1},g_{2}\right)$:
at each step, pick either $g_{1}$ or $g_{2}$ according to the condition
(\ref{eq:or_heuristic}), and update a single bit according to $S_{1}$
or $S_{2}$. 

One possible application for this heuristic is the iterated majority
function. Let $k\in\n$ and $n=3^{k}$. The iterated majority function
$f_{k}$ is recursively defined as follows:
\[
f_{k}=\begin{cases}
\mathrm{maj}\left(x_{1},x_{2},x_{3}\right) & k=1\\
\mathrm{maj}\left(f_{k-1}\left(x_{1},\ldots,x_{3^{k-1}}\right),f_{k-1}\left(x_{3^{k-1}+1},\ldots,x_{2\cdot3^{k-1}}\right),f_{k-1}\left(x_{2\cdot3^{k-1}+1},\ldots,x_{3^{k}}\right)\right) & \text{otherwise,}
\end{cases}
\]
where $\mathrm{maj}:\left\{ -1,1\right\} ^{3}\to\left\{ -1,1\right\} $
returns the most frequent bit in its input. The function $f_{k}$
can be represented by a complete ternary tree of depth $k$, where
the leaves are the input bits $x_{1},\ldots,x_{n}$, and each internal
node has value equal to the majority of its three children. The value
of $f_{k}$ is the value of the root.

The best decision tree complexity of iterated majority is still unknown.
The simplest non-trivial algorithm is as follows: Pick two random
subtrees, and recursively compute their value. If they are equal,
the algorithm terminates. If not, then the third subtree needs to
be computed as well. The expected number of bits queried with this
algorithm is $2.5^{k}$. However, better algorithms exist, which do
not evaluate entire subtrees at once. For example, the algorithm given
by Jayram, Kumar and Sivakumar \cite[Appendix B]{jayram_kumar_sivakumar_two_applications_of_information_complexity}
requires reading only approximately $\left(\frac{13+\sqrt{713}}{16}\right)^{n}\approx2.4813^{k}$
bits on average, under the uniform input. This algorithm recursively
reads one random subtree, but can jump between the two remaining subtrees
until the function value has been determined. Once a single subtree
has been read, the iterated majority turns into either an OR (if the
subtree's value was $1$) or an AND function (if the subtree's value
was $-1$) between the two remaining subtrees. It is then possible
to apply the above heuristic to choose which of the two trees to update.
Using easier-to-compute criterion 
\[
\frac{n-\norm{x_{1}}_{2}^{2}}{1+g_{1}\left(x_{1}\right)}\leq\frac{n-\norm{x_{2}}_{2}^{2}}{1+g_{2}\left(x_{2}\right)}
\]
rather than (\ref{eq:or_heuristic}), we have performed numerical
simulations of this strategy for $k=1,\ldots,9$, yielding an estimated
average number of bits of order $\leq2.45^{k}$. 

\section{Fractional random-turn games and the influence process\label{sec:fractional_random_turn_games}}

Consider the following random-turn two-player game, introduced by
Peres, Schramm, Sheffield and Wilson in \cite{peres_schramm_sheffield_wilson_random_turn_hex}.
Let $f:\left\{ -1,1\right\} ^{n}\to\r$, and let $x\left(0\right)=\left(0,\ldots,0\right)\in\r^{n}$.
At each time $t\in\n$, a coin is flipped. If the result is heads,
player I picks an index $i$ according to some strategy and sets $x_{i}\left(t\right)=1$.
If the result is tails, player II picks an index $i$ and sets $x_{i}\left(t\right)=-1$.
After $n$ turns, we have $x\left(n\right)\in\left\{ -1,1\right\} ^{n}$.
Player I then gains $f\left(x\right)$, while player II loses $f\left(x\right)$.
The goal of each player is to maximize their expected payoff, and
the value of the game is $v=\e\left[f\left(x\left(n\right)\right)\right]$,
when both players play optimally. 
\begin{thm}[Theorem 2.1 in \cite{peres_schramm_sheffield_wilson_random_turn_hex}]
Let $\mu_{1/2}$ be the uniform measure on the hypercube. The value
of a random-turn game is $\e_{x\sim\mu_{1/2}}f\left(x\right)$. Moreover,
any optimal strategy for one of the players is also an optimal strategy
for the other player.
\end{thm}

\begin{lem}[Lemma 3.1 in \cite{peres_schramm_sheffield_wilson_random_turn_hex}]
Let $f:\left\{ -1,1\right\} ^{n}\to\left\{ -1,1\right\} $ be monotone.
Let $T$ be the set of bits that have already been fixed at time $t$,
and let $f_{\mid T}$ be the restriction of $f$ to those bits. Then
a move is optimal if and only if it picks a variable $x_{i}$ such
that $\Inf_{i}\left(f_{\mid T}\right)$ is maximal. 
\end{lem}

In the language of decision trees, attempting to maximize the expected
gain gives rise to a decision tree which always picks a variable with
maximal influence. 

These results can be generalized to fractional query algorithms by
considering \emph{fractional random-turn two-player games}. In these
games, if the coin flip turns up heads, player I picks an index $i$
and sets $x_{i}\left(t\right)=x_{i}\left(t\right)+\eps$, and if the
result is tails, player II picks an index $i$ and sets $x_{i}\left(t\right)=x_{i}\left(t\right)-\eps$,
with the constraint that once a bit reaches the values $\pm1$, it
cannot be picked any more. The following can be proved using the same
proof techniques as in \cite{peres_schramm_sheffield_wilson_random_turn_hex}.
\begin{thm}
Let $\eps=1/M$ for some integer $M>0$. Let $p\in\left[-1,1\right]^{n}\intersect\frac{1}{M}\z^{n}$.
The value of a fractional random-turn game with $x\left(0\right)=p$
is $\e_{x\sim\mu_{p}}f\left(x\right)=f\left(p\right)$. Moreover,
any optimal strategy for one of the players is also an optimal strategy
for the other player. 
\end{thm}

\begin{lem}
Let $\eps=1/M$ for some integer $M>0$, and let $f:\left\{ -1,1\right\} ^{n}\to\left\{ -1,1\right\} $
be monotone. A move is optimal if and only if it selects a bit with
maximal $\partial_{i}f\left(x\left(t\right)\right)$. 
\end{lem}

When both players use the same optimal strategy, $x\left(t\right)$
is an axis-aligned jump process. However, it does not necessarily
minimize the cost $\e\norm{x\left(\tau\right)}_{2}^{2}$. For example,
for decision trees, Simkin \cite{simkin_random_turn_and_richman_games}
showed that the influence-based decision tree for iterated majority
on 9 bits reads $396/64$ bits in expectation, while the best tree
reads only $393/64$ bits. For fractional algorithms, Jacka, Warren
and Windridge \cite{jacka_warren_windridge_minimizing_time_to_a_decision}
considered evaluating the value of $3$-majority, when $X\left(t\right)$
is $3$-dimensional Brownian motion, and only one entry can be moved
at a time. They show that the strategy which always moves the middle
bit takes the least amount of time in expectation. Since $\e B\left(t\right)^{2}=t$
for Brownian motion, this is the same as minimizing the cost of a
fractional query algorithm. However, the derivatives of $3$-majority
$f\left(x_{1},x_{2},x_{3}\right)=\frac{1}{2}\left(x_{1}+x_{2}+x_{3}-x_{1}x_{2}x_{3}\right)$
are given by, for $\left\{ i,j,k\right\} =\left\{ 1,2,3\right\} $,
\[
\partial_{i}f\left(x_{1},x_{2},x_{3}\right)=\frac{1}{2}-\frac{1}{2}x_{j}x_{k}.
\]
When $0<x_{1}<x_{2}<x_{3}$, we have $x_{1}x_{2}<x_{1}x_{3}<x_{2}x_{3}$,
and so $\partial_{3}f>\partial_{2}f>\partial_{1}f$. In this case
the influence-based process would pick the largest entry, not the
middle one. 

\bibliographystyle{plain}
\bibliography{noise_sensitivity_from_fractional_query_algorithm}

\end{document}